\documentclass[cmyk]{article}
\usepackage[margin=1.2in]{geometry}
\usepackage{microtype}
\usepackage[defaultlines=3,all]{nowidow}
\usepackage{amsmath,amsfonts,amssymb}
\usepackage{amsthm}
\usepackage{etoolbox}
\usepackage[ruled,vlined,linesnumbered]{algorithm2e}
\usepackage{thmtools}
\usepackage{thm-restate}

\usepackage{paralist}
\usepackage{multirow}
\usepackage{tabularx}
\usepackage{graphicx}
\usepackage{caption}
\usepackage{subcaption}
\usepackage{authblk}
\usepackage{array}
\newcolumntype{C}[1]{>{\centering\arraybackslash}m{#1}}
\usepackage{tikz}
\usetikzlibrary{calc}
\usetikzlibrary{shapes.geometric}
\tikzset{square/.style={regular polygon,regular polygon sides=4}}
\tikzset{every picture/.style=semithick}

\usepackage[pdfauthor={J.\ Chlebikova, C.\ Dallard}, pdftitle={Colourful components in k-caterpillars and planar graphs}, pdfkeywords={colorful component, caterpillars, planar graphs, complexity, algorithm}]{hyperref}
\hypersetup{
	colorlinks=true,
    linkcolor={red!50!black},
    citecolor={blue!50!black},
    urlcolor={blue!80!black},
bookmarksopen=true,
	bookmarksnumbered,
	bookmarksopenlevel=2,
	bookmarksdepth=3
}
\usepackage[nameinlink,capitalize]{cleveref}
\crefname{subsection}{Subsection}{Subsections}
\usepackage[shortcuts]{extdash} 

\usepackage[backend=biber,style=ieee,sortcites,dashed=false,citestyle=numeric,sorting=anyt,giveninits=true,url=false,isbn=false,eprint=false,mincitenames=1,maxcitenames=2,]{biblatex}
\definecolor{colorOne}{HTML}{9932CC}\definecolor{colorTwo}{HTML}{af0404}\definecolor{colorThree}{HTML}{FF89D2}\definecolor{colorFour}{HTML}{ADFF2F}\definecolor{colorFive}{HTML}{03A9F4}
\definecolor{colorSix}{HTML}{009688}
\definecolor{colorSeven}{HTML}{FFA000}\definecolor{colorEight}{HTML}{FFEE58}\definecolor{colorNine}{HTML}{00FF00}\definecolor{colorTen}{HTML}{A7FFEB}\definecolor{colorEleven}{HTML}{ff1744}\definecolor{colorTwelve}{HTML}{fcb1b1}

\newsavebox{\pinkvertex}
\newsavebox{\lightbluevertex}
\sbox\pinkvertex{\tikz[baseline]{\node[fill=colorThree,draw,circle,anchor=south,inner sep=0.5ex,outer sep=0] {};}}
\sbox\lightbluevertex{\tikz[baseline]{\node[fill=colorTen,draw,circle,anchor=south,inner sep=0.5ex,outer sep=0] {};}}

\usepackage{doi}

\makeatletter
\pretocmd{\NAT@citexnum}{\@ifnum{\NAT@ctype>\z@}{\let\NAT@hyper@\relax}{}}{}{}
\makeatother

\def\NP{\textsf{NP}}

\def\O{\mathcal{O}} \def\CC{\textsc{Colourful Components}}
\def\CP{\textsc{Colourful Partition}}

\usepackage{boxedminipage}
\newcommand{\problemdef}[4][XXXEMPTYLABELXXX]{
	\begin{center}
		\begin{boxedminipage}{.99\textwidth}
			\textsc{{#2}}\ifthenelse{\equal{#1}{XXXEMPTYLABELXXX}}{}{\label{#1}}\\[2pt]
			\begin{tabular}{@{\hspace{5pt}}r p{0.8\textwidth}}
				\textit{Input:}  & {#3}\\
				\textit{Question:}  & {#4}
			\end{tabular}
		\end{boxedminipage}
	\end{center}
}
\usepackage{graphicx}

\theoremstyle{plain}
\newtheorem{theorem}{Theorem}
\newtheorem{lemma}{Lemma}

\theoremstyle{definition}
\newtheorem{definition}{Definition}

\newtheorem{construction}{Construction}
\newtheorem{question}{Question}
\theoremstyle{remark}
\newtheorem{remark}{Remark}

\begin{document}
\title{Colourful components in \texorpdfstring{$k$}{k}-caterpillars and planar graphs\thanks{A preliminary version of the paper was accepted at \textit{IWOCA 2019} \cite{MR3991231}.}}
\date{}

\author[1]{Janka Chleb\'\i kov\'a}
\author[2]{Cl\'ement Dallard\footnote{This work is supported by the Slovenian Research Agency (research project N1-0102).}}
\affil[1]{School of Computing, University of Portsmouth, Portsmouth, United Kingdom\newline
 \texttt{janka.chlebikova@port.ac.uk}}
\affil[2]{UP FAMNIT, University of Primorska, Koper, Slovenia\newline
\texttt{clement.dallard@famnit.upr.si}}

\maketitle

\begin{abstract}
    A connected component of a vertex-coloured graph is said to be \emph{colourful} if all its vertices have different colours.
    By extension, a graph is colourful if all its connected components are colourful.
    Given a vertex-coloured graph $G$ and an integer $p$, the {\CC} problem asks whether there exist at most $p$ edges whose removal makes $G$ colourful and the {\CP} problem asks whether there exists a partition of $G$ into at most $p$ colourful components.
    In order to refine our understanding of the complexity of the problems on trees, we study both problems on $k$\=/caterpillars,
    which are trees with a central path $P$ such that every vertex not in $P$ is within distance $k$ from a vertex in $P$.
    We prove that {\CC} and {\CP} are \NP\=/complete on $4$\=/caterpillars with maximum degree~$3$, $3$\=/caterpillars with maximum degree~$4$ and $2$\=/caterpillars with maximum degree~$5$.
    On the other hand, we show that the problems are linear-time solvable on $1$\=/caterpillars.
    Hence, our results imply two complexity dichotomies on trees: {\CC} and {\CP} are linear-time solvable on trees with maximum degree~$d$ if $d \leq 2$ (that is, on paths), and \NP\=/complete otherwise; {\CC} and {\CP} are linear-time solvable on $k$\=/caterpillars if $k \leq 1$, and \NP\=/complete otherwise.
    We leave three open cases which, if solved, would provide a complexity dichotomy for both problems on $k$\=/caterpillars, for every non-negative integer $k$, with respect to the maximum degree.
    We also show that {\CC} is \NP\=/complete on $5$-coloured planar graphs with maximum degree~$4$ and on $12$-coloured planar graphs with maximum degree~$3$.
    Our results answer two open questions of \citeauthor{MR3968574} mentioned in [\textit{30\textsuperscript{th} Annual Symposium on Combinatorial Pattern Matching}, 2019].
\end{abstract}

\section{Introduction}

In this paper, all graphs are unweighted, undirected and simple (do not contain loops or multiple edges).
A \emph{vertex-coloured graph}, or simply a \emph{coloured graph}, is a graph whose vertices are (not necessarily properly) coloured.
A connected component of a coloured graph is a \emph{colourful component} if all its vertices have different colours.
A graph is said to be \emph{colourful} if all its connected components are colourful.

We focus on two closely related decision problems where a coloured graph and an integer $p$ are given as inputs: (i) the {\CC} problem asks if there exist at most $p$ edges whose removal makes the graph colourful; (ii) the {\CP} problem is to decide if there exists a partition of the vertex set into at most $p$ parts such that each part induces a colourful component in the graph.

The study of both {\CC} and {\CP} was initially motivated by problems raised in comparative genomics.
One key problem is to partition a set of genes into orthologous genes, which are sets of genes in different species that have evolved through speciation events only, \textit{i.e.}\ originated by vertical descent from a single gene in the last common ancestor.
The problem has been modelled as a graph problem where orthologous genes translate into colourful components in the graph~\cite{zheng2011omg,MR3392498}.
Each vertex of the graph represents a gene, and the colour of a vertex illustrates the specie of the corresponding gene.
An edge between two vertices is present in the graph if the two corresponding genes are (sufficiently) similar.
The quality of a partition of a set of genes into orthologous genes can be expressed in different ways.
Minimising the number of similar genes in different subsets of the partition is a well studied variant~\cite{MR2988970,MR3968574,MR1986480,zheng2011omg,MR3836066}, which corresponds to minimising the number of edges between the colourful components (as in the optimisation variant of {\CC}).
Alternatively, one can ask for a partition of minimum size, \textit{i.e}\ a partition containing the minimum number of orthologous genes, or equivalently the minimum number of colourful components~\cite{MR3392498,MR3968574,MR3810330} (as in the optimisation variant of {\CP}).
Another variant, however not studied in this paper, considers an objective function that maximises the number of edges in the transitive closure~\cite{MR3392498,MR3810330,MR3836066}.
See \cref{colourful graph} for an example of a graph partitioned into colourful components.

We give the formal definitions of the two problems considered in this paper.

\problemdef[definition: CC]{\CC}{A vertex-coloured graph $G$, a non-negative integer $p$.}{Are there at most $p$ edges of $G$ whose removal makes $G$ colourful?}

\problemdef[definition: CP]{\CP}{A vertex-coloured graph $G$, a positive integer $p$.}{Is there a partition of the vertex set of $G$ with at most $p$ parts such that each part induces a colourful component in $G$?}

\begin{figure}[htbp]
	\centering
	\begin{tikzpicture}[scale=1.25,rotate=-100]
	\tikzset{every node/.style={draw,circle,fill=gray,inner sep=2.5pt,outer sep=0},
		every label/.append style={rectangle},
		label distance=0pt};

	\node[fill=white] (a) at (0,0) {};
	\node[fill=colorSeven] (b) at ($(a)+(60:1)$) {};
	\node[fill=colorNine] (c) at ($(a)+(120:1)$) {};
	\draw[] (a) -- (b) -- (c) -- (a);

	\coordinate (xa) at ($(a)+(-90:0.2)$);
	\coordinate (xb) at ($(b)+(60:0.2)$);
	\coordinate (xc) at ($(c)+(120:0.2)$);
	\draw [dashed] (xa) to [out=0, in=-30] (xb) to [out=150, in=30] (xc) to [out=210, in=180] (xa);

	\node[fill=colorEleven] (d) at ($(a)+(-60:1)$) {};
	\node[fill=colorSeven] (e) at ($(a)+(-120:1)$) {};
	\node[fill=colorFive] (f) at ($(d)+(-120:1)$) {};
	\node[fill=white] (g) at ($(d)+(-60:1)$) {};
	\draw[] (d) -- (e) -- (f) -- (g) -- (e) (f) -- (d) -- (g) (e) -- (a) -- (d);

	\coordinate (xd) at ($(d)+(60:0.2)$);
	\coordinate (xe) at ($(e)+(120:0.2)$);
	\coordinate (xf) at ($(f)+(-120:0.2)$);
	\coordinate (xg) at ($(g)+(-30:0.2)$);
	\draw [dashed] (xd) to [out=150, in=30] (xe) to [out=180+30, in=150] (xf) to [out=180+150, in=-110] (xg) to [out=180-110, in=-30] (xd);

	\node[fill=colorFive] (h) at ($(e)+(-230:1)$) {};
	\node[fill=colorNine] (i) at ($(h)+(+110:1)$) {};
	\draw[] (e) -- (h) -- (i);

	\coordinate (xh) at ($(h)+(-60:0.2)$);
	\coordinate (xi) at ($(i)+(120:0.2)$);
	\draw [dashed] (xh) to [out=30, in=30] (xi) to [out=-150, in=-150] (xh);
\end{tikzpicture}
 	\caption{A vertex-coloured graph. The dashed parts represent a partition into colourful components. Removing the three edges between the dashed parts yields a colourful graph.}
	\label{colourful graph}
\end{figure}

Observe that, on a tree, there is a solution to {\CC} with $p$ edges if and only if there is a solution to {\CP} with $p+1$ parts.
However, this is not the case on general graphs~\cite{MR3968574}.
Both problems are known to be \NP-complete on subdivided stars~\cite{MR3810330}, trees of diameter at most~$4$~\cite{MR2988970}, and trees with maximum degree~$6$~\cite{MR3968574}.
Trees of diameter at most $4$ are in fact a subclass of $2$\=/caterpillars (defined later), so both problems are \NP-complete on $2$\=/caterpillars (when there is no restriction on the maximum degree).

It is interesting to notice the similarities between {\CC} and the \textsc{Multicut}~\cite{MR3759884,MR3172250} and \textsc{Multi-Multiway Cut}~\cite{MR2323384} problems.
In the \textsc{Multicut} problem, a graph and a set of pairs of vertices are given and the goal is to minimise the number of edges to remove in order to disconnect each pair of vertices.
In the \textsc{Multi-Multiway Cut} problem, a graph and sets of vertices are given and the goal is to minimise the number of edges to remove in order to disconnect all paths between vertices from the same vertex set.
Thus, {\CC} is a special case where such sets of vertices form a partition.

Another well-studied problem dealing with subgraphs of a vertex-coloured graph is the \textsc{Graph Motif} problem~\cite{MR2799278,MR2771171,MR2506799}.
This problem takes a coloured graph and a multiset of colours $M$ (the motif) as input, and the goal is to determine whether there exists a connected subgraph $S$ such that the multiset of colours used by the vertices in $S$ corresponds exactly to $M$.
If $M$ is a set (where each colour appears at most once), then $M$ is said to be colourful.

\paragraph{Preliminaries}
Given a graph $G$, we denote by $V(G)$ the vertex set of $G$ and by $E(G)$ the edge set of $G$.
We assume that a coloured graph $G$ is always associated with a colouring function $c$ from $V(G)$ to a set of colours: for each vertex $u\in V(G)$, $c(u)$ is the colour of the vertex $u$.
The \emph{colour-multiplicity} of $G$ corresponds to the maximum number of occurrences of any colour in the graph.
If $G$ contains at most $\ell$ colours we say that $G$ is an $\ell$-coloured graph.
Given a set of edges $S \subseteq E(G)$, we denote by $G-S$ the graph obtained from $G$ by removing the edges in $S$; if $S$ contains a unique edge $e$, then we may simply write $G-e$.
For a set $X \subseteq V(G)$ of vertices, $G[X]$ denotes the subgraph of $G$ induced by $X$.
The \emph{(open) neighbourhood} of a vertex~$u$, denoted by $N(u)$, is the set of vertices adjacent to~$u$.
The \emph{degree} of a vertex $u$ is the cardinality of its open neighbourhood.
The \emph{closed neighbourhood} of~$u$ is the set $N[u] = N(u) \cup \{u\}$.
A \emph{star} is a tree with at most one vertex with degree more than two.
A \emph{$k$\=/caterpillar} is tree containing a path $P$, called the \emph{backbone}, such that every vertex that is not in $P$ is within distance $k$ from a vertex in $P$.
Note that a $0$-caterpillar is simply a path.

\paragraph{Overview of the results}

\Cref{section trees} focuses on the complexity of {\CC} and {\CP} on trees, and in particular on $k$-caterpillars.
In \cref{hard-caterpillars}, we prove that both problems are \NP-complete on $4$\=/caterpillars with maximum degree~$3$, on $3$\=/caterpillars with maximum degree~$4$, and on $2$\=/caterpillars with maximum degree~$5$.
This answers an open question of \citeauthor{MR3968574}, asked in~\cite{MR3968574}, regarding the complexity of the problems on trees with maximum degree at most~$5$.
Then, in \cref{easy-caterpillars}, we give a linear-time algorithm for {\CC} and {\CP} on $1$\=/caterpillars (without restriction on the maximum degree), even in the case where the backbone induces a cycle.
This result improves the complexity of the previously known quadratic-time algorithm on paths~\cite{MR3810330} and applies to a wider class of graphs.
We then extend our technique for $1$\=/caterpillars to a larger class of graphs.
Of particular interest are the following complexity dichotomies of the problems on trees which we derive from our results.

\begin{restatable}{corollary}{dichotomytreeboundeddegree}\label{dichotomy tree bounded degree}
{\CC} and {\CP} are linear-time solvable on paths, that is, on trees with maximum degree at most $2$, and \NP-complete on trees with maximum degree at least $3$.
\end{restatable}

\begin{restatable}{corollary}{dichotomykcaterpillar}\label{dichotomy k-caterpillar}
    {\CC} and {\CP} are polynomial-time solvable on $k$-caterpillars if $k \leq 1$, and \NP-complete otherwise.
\end{restatable}

In \cref{subcubic graphs}, we consider the complexity of {\CC} on planar graphs with small maximum degree.
It is known that the problem is \NP-complete on $3$\=/coloured graphs with maximum degree~$6$~\cite{MR2988970}, while {\CP} is \NP-complete on $3$\=/coloured $2$\=/connected (planar) graphs with maximum degree~$3$~\cite{MR3968574}.
However, the complexity of {\CC} on $\ell$-coloured graphs with maximum degree at most $5$ was unknown and left as an open question in~\cite{MR3968574}.
We answer that question by showing that {\CC} is \NP-complete on $5$\=/coloured planar graphs with maximum degree~$4$ and on $12$-coloured planar graphs with maximum degree~$3$.
 
\section{Complexity on trees}\label{section trees}
This section is devoted to the study of the complexity of {\CC} and {\CP} on trees.
We first derive \NP\=/completeness results in \cref{hard-caterpillars}, and then present a linear-time algorithm on a superclass of $1$\=/caterpillars in \cref{easy-caterpillars}.
These results imply two complexity dichotomies for {\CC} and {\CP} on trees which we present in \cref{dichotomies}: one dichotomy with respect to the maximum degree and the other with respect to the smallest integer $k$ such that the input tree is a $k$\=/caterpillar, respectively.
\subsection{\texorpdfstring{\NP\=/}{NP-}completeness results}\label{hard-caterpillars}

In this subsection, we show that {\CC} and {\CP} are \NP-complete on $4$\=/caterpillars with maximum degree~$3$, $3$\=/caterpillars with maximum degree~$4$, and $2$-caterpillars with maximum degree~$5$.
We propose a reduction from $3$\=/\textsc{SAT} with at most three occurrence of each variable, denoted by $3{,}3$-\textsc{SAT}, which is proved \NP-complete in \cite{MR1637890}.

\problemdef{\textsc{$3{,}3$-SAT}}{A $3$\=/CNF formula $\phi$ in which each variable occurs at most three times.}{Is there a satisfying assignment for $\phi$?}

If a clause in a $3$\=/CNF formula contains exactly one literal, we can set the value of the corresponding variable deterministically and either reduce the number of clauses or conclude that the formula is unsatisfiable.
Hence, we assume that each clause of the input $3$\=/CNF formula contains at least $2$ literals.
Furthermore we can suppose that each literal appears either one or two times; otherwise the formula can be simplified by setting the variable to the value that satisfies all clauses containing it.

\begin{construction}\label{from phi to tree}
	Let $\phi$ be an instance of $3{,}3$-\textsc{SAT}, that is, a set of $m$ clauses $C_1, C_2, \dots, C_m$ on $n$ variables $x_1, x_2, \dots, x_n$, where each clause contains at most three literals and where each variable appears at most three times.
	We denote by $m_3$ the number of clauses containing three literals and by $m_2$ the number of clauses containing two literals.
	We construct a tree $T$ in which each variable and each clause is represented by a gadget (a connected subtree of $T$).

	For each variable $x_i$ of $\phi$, we construct a \emph{variable gadget:} Firstly, create three vertices labelled $r_{x_i}$, $x_i$, $\bar{x}_i$, and connect $x_i$ and $\bar{x}_i$ to $r_{x_i}$.
	If a clause $C_j$ contains the literal $x_i$, then create a vertex labelled $x_{i,j}$ and connect it to the vertex $x_i$.
	Similarly, if a clause $C_j$ contains the literal $\bar{x}_i$, then create a vertex labelled $\bar{x}_{i,j}$ and connect it to the vertex $\bar{x}_i$.
	Since each literal appears at most two times, $r_{x_i}$ is the root of a tree of depth $2$ with maximum degree~$3$.
	Also, since each literal appears at least once, each leaf of the gadget corresponds to a vertex of the type~$x_{i,j}$, that is, each leaf represents a literal in a clause (see \cref{variable gadget}).

	Then, for each clause $C_j$ of $\phi$, we construct a \emph{clause gadget}.
	If $C_j$ contains only two literals $\ell_{1,j}$ and $\ell_{2,j}$:
	\begin{itemize}
		\item \textit{Clause gadget of type~$A$:}
		      Create two vertices $z_j$, $z'_j$, two vertices labelled $\ell_{1,j}$, $\ell_{2,j}$ representing the literals in $C_j$, and one extra vertex $r_{C_j}$.
		      Then, connect $z_j$ and $z'_j$ to $r_{C_j}$, $\ell_{1,j}$ to $z_j$ and $\ell_{2,j}$ to $z'_j$.
		      The gadget is a tree of depth $2$ with maximum degree~$3$ rooted in $r_{C_j}$ (see \cref{clause gadget type A}).
	\end{itemize}
	If $C_j$ contains three literals $\ell_{1}$, $\ell_{2}$, $\ell_{3}$, then, depending on the desired properties of the final tree $T$, we propose two different kinds of clause gadgets:
	\begin{itemize}
		\item \textit{Clause gadget of type~$B$:}
		      Create four vertices $y_j$, $y'_j$, $z_j$, $z'_j$, three vertices labelled $\ell_{1,j}$, $\ell_{2,j}$, $\ell_{3,j}$ representing the literals in $C_j$, and one extra vertex $r_{C_j}$.
		      Then, add the edges $\{\ell_{1,j}, y_j\}$, $\{\ell_{2,j}, y'_j\}$, $\{\ell_{3,j}, z'_j\}$, $\{y_j, z_j\}$, $\{y'_j, z_j\}$, and the edges $\{z_j, r_{C_j}\}$, $\{z'_j, r_{C_j}\}$.
		      The gadget is a tree of depth $3$ with maximum degree~$3$ rooted in $r_{C_j}$ (see \cref{clause gadget type B}).

		\item \textit{Clause gadget of type~$C$:}
		      Create three vertices $z_j$, $z_j'$, $z_j''$, three vertices labelled $\ell_{1,j}$, $\ell_{2,j}$, $\ell_{3,j}$ representing the literals in $C_j$, and one extra vertex $r_{C_j}$.
		      Then, connect $z_j$, $z_j'$ and $z_j''$ to $r_{C_j}$, $\ell_{1,j}$ to $z_j$, $\ell_{2,j}$ to $z_j'$, and $\ell_{3,j}$ to $z_j''$.
		      The gadget is a tree of depth $2$ with maximum degree~$4$ rooted in $r_{C_j}$ (see \cref{clause gadget type C}).
	\end{itemize}

    \begin{figure}[htb]
	\centering
	\begin{subfigure}[b]{0.45\linewidth}
		\centering
		\begin{tikzpicture}[level/.style={sibling distance = 1.75cm/#1,
				level distance = 0.75cm},scale=1]
	\tikzset{every node/.style={draw,circle,fill=gray,minimum size=5pt,inner sep=2.5pt,outer sep=0},
		every label/.append style={rectangle},
		label distance=2pt};

	\begin{scope}[edge from parent/.style={draw}]\node[fill=colorFive,square,label=above:$\mathstrut r_{x_1}$] at (3,2) {}
		child { node[fill=colorSeven,label=left:$\mathstrut x_1$] {}
				child { node[fill=colorThree,label=below:$\mathstrut x_{1,2}$] {}
					}
				child { node[fill=colorEight,label=below:$\mathstrut x_{1,5}$] {}}
			}
		child { node[fill=colorSeven,label=right:$\mathstrut \bar{x}_1$] {}
				child { node[fill=colorTen,label=below:$\mathstrut \bar{x}_{1,3}$] {}}
			};

\end{scope}
\end{tikzpicture}
 		\caption{Variable gadget of $x_1$}
		\label{variable gadget}
	\end{subfigure}
	\begin{subfigure}[b]{0.45\linewidth}
		\centering
		\begin{tikzpicture}[level/.style={sibling distance = 1.75cm/#1,
				level distance = 0.75cm},scale=1]
	\tikzset{every node/.style={draw,circle,fill=gray,minimum size=5pt,inner sep=2.5pt,outer sep=0},
		every label/.append style={rectangle},
		label distance=2pt};

	\begin{scope}[edge from parent/.style={draw}]
		\node[fill=black,square,label=above:$\mathstrut r_{C_3}$] at (3,2) {}
		child { node[fill=colorTwo,label=left:$\mathstrut z_3$] {}
				child { node[fill=colorTen,label=below:$\mathstrut \ell_{1,3}$] {}}
			}
		child { node[fill=colorTwo,label=right:$\mathstrut z_3'$] {}
				child { node[fill=colorFour,label=below:$\mathstrut \ell_{2,3}$] {}}
			};
\end{scope}

\end{tikzpicture}
 		\caption{Clause gadget of type~$A$\\ of the clause $C_3$}
		\label{clause gadget type A}
	\end{subfigure}
    \vspace{10pt}\\
    \begin{subfigure}[b]{0.45\linewidth}
		\centering
		\begin{tikzpicture}[level/.style={sibling distance = 1.75cm/#1,
				level distance = 0.75cm},scale=1]
	\tikzset{every node/.style={draw,circle,fill=gray,minimum size=5pt,inner sep=2.5pt,outer sep=0},
		every label/.append style={rectangle},
		label distance=2pt};

	\begin{scope}[edge from parent/.style={draw}]\node[fill=white,square,label=above:$\mathstrut r_{C_2}$] at (3,2) {}
		child { node[fill=colorNine,label=left:$\mathstrut z_2$] {}
				child {node[fill=colorEleven,label=left:$\mathstrut y_2$] {}
						child { node[fill=colorThree,label=below:$\mathstrut \ell_{1,2}$] {}}
					}
				child {node[fill=colorEleven,label=right:$\mathstrut y'_2$] {}
						child { node[fill=colorSix,label=below:$\mathstrut \ell_{2,2}$] {}}
					}
			}
		child { node[fill=colorNine,label=right:$\mathstrut z'_2$] {}
				child { node[fill=colorOne,label=below:$\mathstrut \ell_{3,2}$] {}}
			};
\end{scope}

\end{tikzpicture}
 		\caption{Clause gadget of type~$B$\\ of the clause $C_2$}
		\label{clause gadget type B}
	\end{subfigure}
	\begin{subfigure}[b]{0.45\linewidth}
		\centering
		\begin{tikzpicture}[level/.style={sibling distance = 1.75cm/#1,
				level distance = 0.75cm},scale=1]
	\tikzset{every node/.style={draw,solid,circle,fill=gray,minimum size=5pt,inner sep=2.5pt,outer sep=0},
every label/.append style={rectangle},
		label distance=2pt};

	\begin{scope}[level/.style={sibling distance = 1.25cm/#1, level distance = 0.95cm}, edge from parent/.style={draw}]\node[fill=white,square,label=above:$\mathstrut r_{C_2}$] at (3,2) {}
		child { node[fill=colorNine,label=left:$\mathstrut z_2$] {}
				child { node[fill=colorThree,label=below:$\mathstrut \ell_{1,2}$] {}}
			}
		child { node[fill=colorNine,label=right:$\mathstrut z_2'$] {}
				child { node[fill=colorSix,label=below:$\mathstrut \ell_{2,2}$] {}}
			}
		child { node[fill=colorNine,label=right:$\mathstrut z_2''$] {}
				child { node[fill=colorOne,label=below:$\mathstrut \ell_{3,2}$] {}}};
\end{scope}

\end{tikzpicture}
 		\caption{Clause gadget of type~$C$\\ of the clause $C_2$}
		\label{clause gadget type C}
	\end{subfigure}
	\caption{An example of variable and clause gadgets used in \cref{from phi to tree} given a $3{,}3$-\textsc{SAT} instance $\phi$.
	The variable $x_1$ appears twice as a positive literal (in $C_2$ and $C_5$) and once as a negative literal (in $C_3$) in $\phi$.
	The clause $C_2$ (of size~$3$) in $\phi$, which can be represented with a clause gadget of type~$B$ or $C$, contains $x_1$ as a positive literal (pink vertex \usebox{\pinkvertex} corresponding to vertex $x_{1,2}$).
	The clause $C_3$ (of size~$2$) in $\phi$, which must be represented with a clause gadget of type~$A$, contains $x_1$ as a negative literal (light blue vertex \usebox{\lightbluevertex} corresponding to vertex $\bar{x}_{1,3}$).
	Squared vertices have a colour that appears only once in the obtained tree~$T$.}\label{figure tree}
\end{figure}

	We now explain which clause gadgets must be used, and how clause and variable gadgets must be connected together, so that the final tree $T$ has the required properties.
	First, create a variable gadget for each variable.
	Also, for each clause containing two literals only, create a clause gadget of type~$A$.
	Then, apply one of the following options:
	\begin{itemize}
		\item \textit{To get $T$ as a $4$\=/caterpillar with maximum degree~$3$:} Create a clause gadget of type~$B$ for each clause containing three literals.
		Then, create a central path with $n+m_3+m_2$ new vertices, and connect all $r_{x_i}$ and $r_{C_j}$ vertices to distinct vertices of the central path.

		\item \textit{To get $T$ as a $3$\=/caterpillar with maximum degree~$4$:} Create a clause gadget of type~$B$ for each clause containing three literals.
		Then, connect all $r_{x_i}$ and $r_{C_j}$ vertices together to create a central path.
		Alternatively, create a clause gadget of type~$C$ for each clause containing three literals, create a central path with $n+m_3+m_2$ new vertices and connect all $r_{x_i}$ and $r_{C_j}$ vertices to distinct vertices of the central path.

		\item \textit{To get $T$ as a $2$\=/caterpillar with maximum degree~$5$:} Create a clause gadget of type~$C$ for each clause containing three literals.
		Then, connect all $r_{x_i}$ and $r_{C_j}$ vertices together to create a central path.
	\end{itemize}

	In all cases, the central path corresponds to the backbone of $T$.
	We set the root $r$ of $T$ such that it belongs to the backbone and has minimum degree, that is, two, three or four children if $T$ has maximum degree $3$, $4$ or $5$, respectively.

	Lastly, we assign a colour to each vertex in $T$.
	For each variable gadget of a variable $x_i$, let $c(x_i) = c(\bar{x}_i)$ be a new colour.
	Also, for each vertex $\widetilde{x}_{i,j} \in \{ x_{i,j}, \bar{x}_{i,j} \}$, let $c(\widetilde{x}_{i,j})$ be a new colour.
	Then, for each clause gadget of clause $C_j$, if the literal $\ell_{k,j} = x_i$ in $C_j$, then set $c(\ell_{k,j}) := c(x_{i,j})$, but if $\ell_{k,j} = \bar{x}_i$, then set $c(\ell_{k,j}) := c(\bar{x}_{i,j})$.
	If a clause gadget is of type~$B$, then let $c(y_j) = c(y'_j)$ and $c(z_j) = c(z'_j)$ be two new colours.
	If a clause gadget is of type~$C$, then let $c(z_j) = c(z'_j) = c(z''_j)$ be a new colour.
	If a clause gadget is of type~$A$, then let $c(z_j) = c(z'_j)$ be a new colour.
    Finally, each remaining uncoloured vertex (root vertex of a variable or clause gadget, or vertex of the central path) is assigned a new (unused) colour.
	If $T$ does not contain clause gadgets of type~$C$, then its colour-multiplicity is~$2$, otherwise it is~$3$.
	An example of variable and clause gadgets with coloured vertices is given in \cref{figure tree}.
\end{construction}

Note that \cref{from phi to tree} can be done in polynomial time.

\begin{theorem}\label{np-hard caterpillars}
	{\CC} and {\CP} are \NP-complete on coloured :
	\begin{itemize}
		\item $2$\=/caterpillars with maximum degree~$5$ and colour-multiplicity $3$,
		\item $3$\=/caterpillars with maximum degree~$4$ and colour-multiplicity $2$, and
		\item $4$\=/caterpillars with maximum degree~$3$ and colour-multiplicity $2$.
	\end{itemize}
\end{theorem}
\begin{proof}
	Obviously, the problems are in \NP.
	Since the problems are equivalent on trees, we prove the statement for {\CC}.
	Let $\phi$ be an instance of $3{,}3$-\textsc{SAT} with $m$ clauses and $n$ variables.
	We denote by $m_3$ the number of clauses containing $3$ literals and by $m_2$ the number of clauses containing $2$ literals.
	We transform $\phi$ into a coloured tree $T$ as described in \cref{from phi to tree}.
	Note that $T$ is either a $2$\=/caterpillar with maximum degree~$5$, a $3$\=/caterpillar with maximum degree~$4$ or a $4$\=/caterpillar with maximum degree~$3$.
	The colour-multiplicity of $T$ depends on the use of clause gadgets of type~$C$, and therefore is $3$ for a $2$\=/caterpillar with maximum degree~$5$ but can only be $2$ for a $3$\=/caterpillar with maximum degree~$4$ or a $4$\=/caterpillar with maximum degree~$3$.
	We claim that there exists a satisfying assignment for $\phi$ if and only if there is a set of exactly $n+2m_3+m_2$ edges in $T$ whose removal makes $T$ colourful.

	Let $\beta$ be a satisfying assignment for $\phi$.
	We define the set of edges $S$ as follows:
	\begin{itemize}
		\item For each variable $x_i$, the set $S$ contains the edge $\{r_{x_i},x_i\}$ if $x_i = True$ in $\beta$, or $\{r_{x_i},\bar{x}_i\}$ if $x_i = False$ in $\beta$ .

		\item For each clause $C_j$ with three literals:
		      \begin{itemize}
			      \item If the clause gadget is of type~$B$, then the set $S$ contains two edges: one from the path between $y_j$ and $y'_j$, and one from the path between $z_j$ and $z'_j$.
			            These edges are chosen in such a way that, in $G-S$, the leaf $\ell_{k,j}$ which belongs to the same connected component as the vertex $r_{C_j}$ corresponds to (one of) the literal(s) satisfying the clause $C_j$ in $\beta$.

			      \item If the clause gadget is of type~$C$, then the set $S$ contains two edges incident to the vertex $r_{C_j}$ such that $z_j$, $z'_j$, and $z''_j$ are in different connected component in $G-S$.
			            These edges are chosen so that, in $G-S$, the leaf $\ell_{k,j}$ which belongs to the same connected component as the vertex $r_{C_j}$ corresponds to (one of) the literal(s) satisfying the clause $C_j$ in $\beta$.
		      \end{itemize}

		\item For each clause $C_j$ with two literals, which is represented by a clause gadget of type~$A$, the set $S$ contains either the edge $\{r_{C_j}, z_j\}$ or the edge $\{r_{C_j},z'_j\}$.
    Again, this edge is chosen in order that, in $G-S$, the leaf $\ell_{k,j}$ which belongs to the same connected component as the vertex $r_{C_j}$ corresponds to (one of) the literal(s) satisfying the clause $C_j$ in $\beta$.
	\end{itemize}
	Clearly, the set $S$ contains $n+2m_3+m_2$ edges.
	We denote by $F$ the forest $T-S$, and by $T'$ the connected component in $F$ containing the root $r$ of $T$.
	Obviously, two vertices of the same colour from a same variable gadget do not both belong to a same connected component of $F$, and the same holds for a clause gadget.
	Also, note that two vertices of different variable gadgets do not have the same colour, and similarly for vertices of different clause gadgets.
	Lastly, observe that two vertices of two different gadgets belong to the same connected component if and only if they are connected through the backbone, which is in $T'$.
	Thus, if there exist two vertices of the same colour in a same connected component of $F$, one is from a variable gadget and the other one from a clause gadget; so both vertices belong to $T'$.
	Without loss of generality, consider $x_{i,j}$ from the variable gadget of $x_i$ and $\ell_{k,j}$ from the clause gadget of $C_j$, such that $x_{i,j}, \ell_{k,j} \in T'$.
	To derive a contradiction, suppose that $c(x_{i,j}) = c(\ell_{k,j})$.
	Note that the literal represented by $\ell_{k,j}$ is $x_{i,j}$, otherwise the two vertices would not have the same colour.
	Since $\ell_{k,j}$ is in $T'$, it is connected to the vertex $r_{C_j}$ of the clause gadget, hence $\ell_{k,j}$ satisfies the clause $C_j$.
	Therefore, the variable $x_i = True$ in $\beta$.
	By construction, this implies that the edge $\{r_{x_i},x_i\}$ belongs to $S$, and that the subtree $T_{x_i}$, containing $x_{i,j}$, is not part of $T'$, which is a contradiction.

	Let $S$ be a solution to {\CC} for $T$.
	Observe that $S$ must contain at least one edge per variable gadget to put the vertices $x_i$ and $\bar{x}_i$, which have the same colour, into different connected components.
  Besides, since each clause gadget of type~$B$ contains two pairs of vertices with the same colours (the pairs $z_i,z'_i$ and $y_i,y'_i$), and each clause gadget of type~$C$ contains three vertices of the same colour ($z_i$, $z'_i$ and $z''_i$), $S$ must contain at least two edges per clause gadget of type~$B$ or type~$C$.
  Lastly, as each clause gadget of type~$A$ contains two vertices with the same colour, $S$ must contain at least one edge per clause gadget of type~$A$.
  Putting everything together, we get that $|S| \geq n+2m_3+m_2$.
  Therefore, if we suppose that $|S| = n+2m_3+m_2$, then $S$ contains exactly one edge per variable gadget, exactly two edges per clause gadget of type~$B$ or type~$C$, and exactly one edge per clause gadget of type~$A$.
  Note that $S$ does not contain any edge from the backbone of $T$.
	Let $T'$ be the connected component of $T-S$ containing the root $r$.
	Notice that, for each variable gadget, either $x_i$ or $\bar{x}_i$ belongs to $T'$, but not both.
	Also, for each clause gadget, exactly one leaf $\ell_{k,j}$ belongs to $T'$.
	We construct an assignment $\beta$ of $\phi$ such that, for each variable gadget, if $\{r_{x_i},x_i\} \in S$, then we set $x_i := True$ in $\beta$, and if $\{r_{x_i},\bar{x}_i\} \in S$, then we set $x_i := False$ in $\beta$.
	Suppose for a contradiction that there is a clause $C_j$ which is not satisfied in $\phi$ with regard to $\beta$.
	Consider the leaf $\ell_{k,j} \in T'$ from the gadget clause of $C_j$, and assume without loss of generality that $\ell_{k,j} = x_{i}$.
	If $C_j$ is not satisfied, then the variable $x_{i} := False$ in $\beta$.
	This means that $S$ contains the edge $\{r_{x_i},\bar{x}_{i}\}$, but not the edge $\{r_{x_i},x_{i}\}$, and thus $x_{i,j} \in T'$.
	However, since $c(x_{i,j}) = c(\ell_{k,j})$, then $S$ is not a solution for $T$, a contradiction.
\end{proof}

\Citeauthor{MR3968574} asked in~\cite{MR3968574} what is the complexity of {\CC} and {\CP} on trees of maximum degree $d$ for $3 \leq d \leq 5$.
\Cref{np-hard caterpillars} answers the question completely, since both problem are \NP\=/complete on $4$-caterpillars with maximum degree $3$. 
\subsection{A linear-time algorithm for \texorpdfstring{$1$\=/}{1-}caterpillars, and more}\label{easy-caterpillars}

In this subsection, we show that {\CC} and {\CP} can be solved in linear time on coloured $1$-pseudocaterpillars, without restriction on the maximum degree.
\begin{definition}[Pseudocaterpillar]
    A \emph{cyclic $1$-caterpillar} is a connected graph with a unique cycle $B$, called the backbone, such that for any $e \in E(B)$, the graph $G-e$ is a $1$-caterpillar.
    A \emph{$1$-pseudocaterpillar} is either a $1$-caterpillar or a cyclic $1$-caterpillar.
    For the sake of simplicity, we refer to $1$-pseudocaterpillars as \emph{pseudocaterpillars}.
\end{definition}

If a coloured pseudocaterpillar is not colourful, then it contains either one or at least two colours that appear more than once.
We deal with these two cases independently in the following lemmas.

\begin{lemma}\label{exactly one colour appears at least twice}
	{\CC} and {\CP} can be solved in linear time on coloured pseudocaterpillars where at most one colour appears at least twice.
\end{lemma}
\begin{proof}
	Let $G=(V,E)$ be a coloured pseudocaterpillar such that there is at most one colour appearing at least twice in~$G$.
	If each colour appears at most once, then $G$ is colourful.
	So let $\gamma$ be a colour that appears $n_\gamma \geq 2$ times in~$G$.

	Notice that if the backbone is a path, then an optimal solution to {\CP} with $n_\gamma$ parts can be found in linear time.
	Clearly, the $n_\gamma-1$ edges with endpoints in different parts form an optimal solution to {\CC} and can also be obtained in linear time.

	Suppose that the backbone induces a cycle and let $b_\gamma$ be the number of vertices of colour $\gamma$ on the backbone.
	It is easy to see that if $b_\gamma \leq 1$, the removal of $n_\gamma - 1$ edges is necessary and sufficient.
	Assume now that $b_\gamma \geq 2$.
	Since at least two vertices of colour $\gamma$ belong to the backbone (which is a cycle), any solution to {\CC} must contain at least $2$ edges from the backbone.
	Remove any edge $e$ from the backbone and let $G'$ be the resulting graph.
	Since the backbone of $G'$ is a path, $G'$ is a tree, and thus any optimal solution $S'$ to {\CC} on $G'$ contains exactly $n_\gamma-1$ edges.
	Of course, $S' \cup \{e\}$ is a solution for $G$ that can be obtained in linear time and contains at most $n_\gamma$ edges.
	We claim that such a solution is optimal.
	Let $S$ be a solution to {\CC} for $G$.
	To obtain a contradiction, suppose that $|S| < n_\gamma$.
	Let $e \in S$ be an edge on the backbone (there exist at least two such edges in $S$), and let $G'$ be the graph obtained by removing $e$ from~$G$.
	Of course, $G'$ is connected and has a path as backbone.
	Since $S$ is a solution to {\CC} for $G$, then $S' = S \setminus \{e\}$ is a solution for $G'$ of size $|S'| = |S|-1 < n_\gamma-1$, which is impossible since $G$ is connected and contains $n_\gamma$ vertices of colour~$\gamma$.

	Clearly, the partition of the vertices defined by each connected component is an optimal solution to {\CP} since each component contains exactly one vertex of colour~$\gamma$.
\end{proof}

We now deal with the case where the input pseudocaterpillar contains at least two colours that appear more than once.
We consider the vertices of the backbone as internal vertices of stars whose leaves are the adjacent vertices with degree $1$.
When the backbone is a path, we assume it is of minimum length.
This in particular implies that, if the input pseudocaterpillar contains at least $3$ vertices, then there is a bijection between vertices of the backbone and vertices with degree at least $2$, and conversely a bijection between leaves of the stars and vertices with degree $1$.

\begin{definition}[{$C[u]$}]
	Given a coloured pseudocaterpillar $G$ with backbone $B$ and a vertex $u \in V(B)$, we denote by $C[u]$ the multiset of colours used by all the vertices of the star having $u$ as internal vertex.
	Note that using a breadth-first search, the value of $C[u]$ can be computed in linear time for all $u \in V(B)$.
\end{definition}

\begin{remark}\label{preprocessing}
	Consider a coloured pseudocaterpillar $G$ with backbone $B$ and a vertex $u \in V(B)$.
	If $C[u]$ contains a colour $\gamma$ more than once, then at least one vertex with colour $\gamma$ in the corresponding star is a leaf which must belong to a different colourful component than $u$.
	Hence, $G$ can be preprocessed in such a way that, for each such leaf, we add its adjacent edge to a set $S_p$.
	This procedure is repeated until there is no such leaf in $G-S_p$.
	At the end of the preprocessing, each star of $G-S_p$ is a \emph{colourful star}, that is, only contains vertices with different colours.
	See an example of such a preprocessing in \cref{preprocessed cyclic pseudocaterpillar}.
\end{remark}

A \emph{coloured pseudocaterpillar with colourful stars} is a coloured pseudocaterpillar whose all stars are colourful.

\begin{figure}[hbpt]
	\centering
	\begin{tikzpicture}[scale=0.95]
	\tikzset{every node/.style={draw,circle,fill=gray,inner sep=2.5pt,outer sep=0},
		every label/.append style={rectangle},
		label distance=10pt};

	\node[fill=colorSeven] (c1) at (0,0) {};
	\node[fill=colorNine] (c2) at ($(c1) + (45:1)$) {};
	\node[fill=colorEleven] (c3) at ($(c2) + (10:1)$) {};
	\node[fill=colorSeven] (c4) at ($(c3) + (-10:1)$) {};
	\node[fill=colorEight] (c5) at ($(c4) + (-45:1)$) {};
	\node[fill=colorFive] (c6) at ($(c5) + (-135:1.25)$) {};
	\node[fill=colorEleven] (c7) at ($(c6) + (-180:1.5)$) {};

	\draw[] (c1) -- (c2);
	\draw[] (c2) -- (c3);
	\draw[] (c3) [] -- (c4);
	\draw[] (c4) -- (c5);
	\draw[] (c5) [] -- (c6);
	\draw[] (c6) [] -- (c7);
	\draw[] (c7) [] -- (c1);

	\node[fill=colorOne] (d1) at ($(c1) + (160:1)$) {};
	\node[fill=colorOne] (e1) at ($(c1) + (200:1)$) {};
	\draw[] (c1) [dotted] -- (d1);
	\draw[] (c1) -- (e1);

	\node[fill=colorEight] (d2) at ($(c2) + (115.5:1)$) {};
	\draw[] (c2) -- (d2);

	\node[fill=colorFive] (d4) at ($(c4) + (97.5:1)$) {};
	\node[fill=colorOne] (e4) at ($(c4) + (67.5:1)$) {};
	\node[fill=colorSeven] (f4) at ($(c4) + (37.5:1)$) {};
	\draw[] (c4) -- (d4);
	\draw[] (c4) -- (e4);
	\draw[dotted] (c4) -- (f4);

	\node[fill=colorEleven] (d5) at ($(c5) + (0:1)$) {};
	\draw[] (c5) -- (d5);

	\node[fill=colorNine] (d6) at ($(c6) + (-47.5:1)$) {};
	\node[fill=colorSeven] (e6) at ($(c6) + (-87.5:1)$) {};
	\draw[] (c6) -- (d6);
	\draw[] (c6) -- (e6);

	\node[fill=colorEight] (d7) at ($(c7) + (-180+22.5:1)$) {};
	\node[fill=colorNine] (e7) at ($(c7) + (-180+52.5:1)$) {};
	\node[fill=colorOne] (f7) at ($(c7) + (-180+82.5:1)$) {};
	\node[fill=colorEight] (g7) at ($(c7) + (-180+112.5:1)$) {};
	\draw[dotted] (c7) -- (d7);
	\draw[] (c7) -- (e7);
	\draw[] (c7) -- (f7);
	\draw[] (c7) -- (g7);

\end{tikzpicture}
 	\caption{A $6$\=/coloured pseudocaterpillar. Dotted edges belong to all solutions to {\CC} (up to isomorphism) and are removed from the graph in the preprocessing.}
	\label{preprocessed cyclic pseudocaterpillar}
\end{figure}

\begin{definition}[Bad path]
	A path $P$ in a coloured graph $G$ between two distinct vertices $u$ and $v$ of colour $\gamma$ is called a \emph{$\gamma$-bad path}.
	A \emph{bad path} is a \emph{$\gamma$-bad path} for some colour~$\gamma$.
\end{definition}

Note that a graph is colourful if and only if it does not contain a bad path.

We say that a solution to $\CC$ is \emph{optimal} if it is a solution of minimum size.

\begin{lemma}\label{at least two colours appear twice}
	Let $G$ be a coloured pseudocaterpillar with colourful stars such that at least two colours appear at least twice in~$G$.
	Then there exists an optimal solution $S$ to {\CC} for $G$ such that $S \subseteq E(B)$, where $B$ is the backbone of~$G$.
\end{lemma}
\begin{proof}
    Note that, since every star of $G$ is colourful and at least two colours appear at least twice in $G$,	the backbone $B$ must contain at least $2$ vertices.
	For simplicity, let $x_0, \dots, x_t$ be the vertices in $V(B)$ ordered such that $\{x_i, x_{i+1}\} \in E(B)$ for all $i \in \{0,\dots,t-1\}$, and $\{x_t, x_0\} \in E(B)$ when $B$ is a cycle.

	It is easily observed that if the backbone $B$ is a path, then there exists an optimal solution $S$ to {\CC} for $G$ such that $S \subseteq E(B)$.
	For instance, one can construct such an optimal solution $S$ as follows:
	let $S = \emptyset$ and $Col = \emptyset$;
	traverse $B$ starting at $x_0$;
	each time a vertex $x_i$ is visited ($x_0$ being the first visited vertex), if $C[x_i] \cap Col \neq \emptyset$, then add the edge $\{x_{i-1},x_i\}$ to $S$ and set $Col := C[x_i]$;
	otherwise, set $Col := Col \cup C[x_i]$.
	In particular, all along the algorithm, $Col$ contains the colours used by the vertices of the stars whose internal vertices are in the same connected component as the newly visited vertex $x_i$ in $G-S$.
	Clearly, once all the vertices of $B$ have been visited, the set $S$ is an optimal solution for~$G$.

	So we can focus on the case where $B$ is a cycle.
	Let $S$ be an optimal solution to {\CC} for $G$ maximising $|S \cap E(B)|$.
	We assume that $S \nsubseteq E(B)$, otherwise the statement is true.
	We consider two cases.

    Suppose first that there exists an edge $f \in S \cap E(B)$.
    Then the backbone $B'$ of $G' = G - \{f\}$ is a path, and thus there exists an optimal solution $S'$ for $G'$ such that $S' \subseteq E(B')$.
	Obviously, $S' \cup \{f\}$ is a solution to {\CC} for $G$.
	We claim that $S' \cup \{f\}$ is an optimal solution for~$G$.
	If not, then $|S| \leq |S'|$.
	This implies that $S \setminus \{f\}$ is a solution for $G'$ of size $|S|-1 < |S'|$, contradicting the optimality of~$S'$.
	So $S' \cup \{f\}$ is an optimal solution for $G$.
	However, as $S' \cup \{f\} \subseteq E(B)$, this contradicts the choice of $S$.

	Suppose now that $S \cap E(B) = \emptyset$.
	Recall that at least two colours $\alpha$ and $\beta$ appear at least twice in~$G$, and thus there exists at least one $\alpha$-bad path and at least one $\beta$-bad path in~$G$.
	Since $S \cap E(B) = \emptyset$, no edge in $S$ can belong to an $\alpha$-bad path and a $\beta$-bad path.
	Thus, $S$ contains two distinct edges $e = \{x_i,w\}$ and $f = \{x_j,y\}$ with $x_i,x_j \in V(B)$ and $w,y \notin V(B)$, and such that $c(w) = \alpha$ and $c(y) = \beta$ (note that we may have $x_i = x_j$).
	If $x_i = x_j$, then the set $S' = (S \setminus \{e,f\}) \cup \{e',f'\}$ where $e'$ and $f'$ are the two distinct edges of $B$ incident to $x_i (= x_j)$ is also a solution for $G$.
	Besides, either $|S'| < |S|$, which contradicts the optimality of $S$, or $|S'| = |S|$ and $|S' \cap E(B)| > |S \cap E(B)|$, which contradicts the choice of $S$.
	So we conclude that $x_i \neq x_j$.
	Since $S$ is an optimal solution, there exist exactly two distinct vertices $w',y' \in V(G) \setminus \{w,y\}$ in the same connected component as $x_i$ and $x_j$ in $G-S$ such that $c(w') = c(w) = \alpha$ and $c(y') = c(y) = \beta$.
	Denote by $x_k$ the internal vertex of the star containing $w'$ and by $x_\ell$ the internal vertex of the star containing $y'$.
	Note that we may have $x_k = x_{\ell}$; if $x_k \neq x_\ell$, then we may have $x_i = x_\ell$ or $x_j = x_k$.
	Let $P_k$ be a path in $B$ between $x_i$ and $x_j$ such that $x_k \in V(P_k)$, and $e' \in E(P_k)$ the edge incident to $x_i$.
	Similarly let $P_\ell$ be a path in $B$ between $x_i$ and $x_j$ such that $x_{\ell} \in V(P_\ell)$, and $f' \in E(P_\ell)$ the edge incident to $x_j$.
    If $x_i = x_\ell$ or $x_j = x_k$, then the paths $P_k$ and $P_\ell$ are not uniquely defined.
	In this case, we make sure to choose $P_k$ and $P_\ell$ such that one is not a subgraph of the other (which is always possible since $x_i \neq x_j$ and $B$ is a cycle).
	We claim that $S' = (S \setminus \{e,f\}) \cup \{e',f'\}$ is an optimal solution for~$G$.
	First, observe that, for any colour $\gamma \notin \{\alpha,\beta\}$, there are no $\gamma$-bad paths in any connected component of $G-S'$.
    Indeed, if there were an $\alpha$-bad path or a $\beta$-bad path $P$ in $G-S'$, then, necessarily, $w$ or $y$ would be an endpoint of $P$.
	However, the fact that $e', f' \in S'$ guarantees that $x_i$ and $x_k$ are in two different connected component of $G-S'$, and the same holds for $x_j$ and $x_{\ell}$.
	In particular, $w$ is the only vertex of colour $\alpha$ in its connected component of $G-S'$ and, similarly, $y$ is the only vertex of colour $\beta$ in its connected component of $G-S'$.
	Hence, $S'$ is a solution for $G$ such that $|S'| \leq |S|$ and $|S' \cap E(B)| > |S \cap E(B)|$, a contradiction with the choice of $S$.

	Thus, we conclude that $S \subseteq E(B)$.
\end{proof}

In order to obtain a linear-time algorithm, we define a special type of bad paths, called \emph{critical bad paths}.

\begin{definition}[Critical bad path]
	Let $G$ be a coloured pseudocaterpillar with backbone $B$ and colourful stars.
	A \emph{$\gamma$-critical bad path} $P$ is a maximal path in $B$ such that for all internal vertex $w$ of $P$ it holds $\gamma \notin C[w]$.
	A \emph{critical bad path} is a $\gamma$-critical bad path for some colour~$\gamma$.
	Note that two $\gamma$-critical bad paths do not have common edges.
\end{definition}

\begin{remark}\label{critical bad paths are enough}
	Let $G$ be a coloured pseudocaterpillar with colourful stars such that at least two colours appear twice.
	Let $B$ be the backbone of~$G$.
	Since $G$ has colourful stars, each bad path, and therefore each critical bad path, contains at least one edge in $B$.
  It is easily observed that if a set of edges $S \subseteq B$ contains at least one edge from each critical bad path, then $S$ contains at least one edge of each bad path, and thus $S$ is a solution to {\CC}.
  Moreover, since the existence of a bad path implies the existence of a critical bad path, if $S$ is of minimum size, then according to \cref{at least two colours appear twice} $S$ is an optimal solution to {\CC}.
\end{remark}

We observe that, although the number of bad paths can be quadratic in the number of vertices, there is only a linear number of critical bad paths, and thus a linear number of edges from the backbone that have to be removed from the input graph to make it colourful.
This observation together with \cref{critical bad paths are enough} allows us to develop a linear-time algorithm on coloured pseudocaterpillars.

The idea of the algorithm is to define a circular-arc graph $H$ (an intersection graph of a collection of arcs on the circle) based on the critical bad paths of~$G$.
\Citeauthor{MR1143909} showed that a minimum clique cover $\mathcal{Q}$ of $H$, which is a partition of the vertex set into a minimum number of cliques, can be obtained in linear time~\cite{MR1143909}.
We show that $\mathcal{Q}$ can then be translated back into an optimal solution to {\CC} and {\CP} for $G$ in linear time.

\begin{algorithm}[htb]
	\KwIn{An $\ell$-coloured pseudocaterpillar with colourful stars $G$ and backbone $B$.}
	\KwOut{A multiset of ordered pairs of vertices.}
	\tcp{Initialisation}
	let $A$ be an empty multiset of ordered pairs of vertices\;
	let $L$ be an array of length $\ell$ initialised at $NULL$\;
	let $x_0, \dots, x_p \in V(B)$ s.t. $\forall i \in \{1,\dots,p-1\}$, $\{x_i,x_{i+1}\} \in E(B)$\;
	let $end := NULL$\;
	let $proceed := True$\;
	let $i := 0$\;
	\tcp{Main procedure}
	\While(\tcp*[f]{$\O(n)$}){$proceed$}{
		\If(\tcp*[f]{no more critical bad path}){$x_i = end$}{
			$proceed := False$\tcp*{$B$ is a cycle}\label{second visit of end stop}
		}
		\ForEach{$\gamma \in C[x_i]$}{
			\If(\tcp*[f]{first time colour $\gamma$ is seen}){$L[\gamma] = NULL$}{\label{first time}
				$end := x_i$\;\label{vertex end updated}
			}
			\ElseIf(\tcp*[f]{critical bad path detected}){$L[\gamma] \neq x_i$}{\label{bad path detected}
				add $(L[\gamma],x_i)$ to $A$\;\label{arc added}
			}
			$L[\gamma] := x_i$\;
		}
		\If{$i = p$ and $\{x_p,x_0\} \notin E(B)$}{
			$proceed := False$\tcp*{$B$ is a path}\label{end of backbone stop}
		}
		$i := (i+1) \mod (p+1)$\;
	}
	return $A$\;
	\caption{From coloured pseudocaterpillar to ordered pairs.}\label{from pseudocaterpillar to circular-arc graph}
\end{algorithm}

\begin{lemma}\label{bad path equiv arc}
	Let $G$ be a coloured pseudocaterpillar with colourful stars, and $A$ be the multiset of pairs returned by \cref{from pseudocaterpillar to circular-arc graph}.
	Then there is a bijection between the set of critical bad paths in $G$ and the multiset $A$.
\end{lemma}
\begin{proof}
	Let $B$ be the backbone of $G$.
	A critical-bad path $P$ between two distinct vertices $x_j,x_i \in V(B)$ is detected in \cref{from pseudocaterpillar to circular-arc graph} at \cref{bad path detected}, when the algorithm detects that there exist a colour $\gamma \in C[x_i]$ such that $L[\gamma] \notin \{ NULL, x_i\}$.
	When $x_i$ is considered in the algorithm, the pair $(L[\gamma],x_i)$ is added to $A$ at \cref{arc added}, and since $L\gamma] = x_j$, then $(x_j,x_i) \in A$.
	Note that $P$ is a $\gamma$-critical bad path, since $x_j = L[\gamma]$ corresponds to the last visited vertex, before $x_i$, such that $\gamma \in C[x_j]$.
	Therefore, the arc with endpoints $(x_j,x_i) \in A$ corresponds to the critical bad path $P$ between $x_j$ and $x_i$ in~$G$.

	An ordered pair $(x_j,x_i)$ in $A$ refers to two vertices $x_i$ and $x_j$ in $V(B)$.
	If such a pair exists, then there exists a colour $\gamma$ such that $\gamma \in C[x_j] \cap C[x_i]$.
	Observe that when $(x_j,x_i)$ is added to $A$, it holds $L[\gamma] = x_j$.
	This implies that the path $P$ between $x_j$ and $x_i$ in $B$, with respect to the order in which the vertices in $V(B)$ are visited, does not contain an internal vertex $x_k$ such that $\gamma \in C[x_k]$.
	Thus, $P$ is a $\gamma$-critical bad path and it corresponds to the pair $(x_j,x_i)$ in~$A$.
\end{proof}

\begin{lemma}\label{algo is linear}
	\Cref{from pseudocaterpillar to circular-arc graph} runs in linear time.
\end{lemma}
\begin{proof}
	Let $G$ be a coloured pseudocaterpillar with colourful stars and backbone $B$, and $A$ the multiset of ordered pairs obtained by \cref{from pseudocaterpillar to circular-arc graph} with input~$G$.

	In \cref{from pseudocaterpillar to circular-arc graph}, when a colour is seen for the first time at \cref{first time}, the internal vertex $u$ of the star is stored in the variable $end$.
	If the backbone induces a cycle, then the second time that the vertex $end$ is considered in the main loop the algorithm sets the variable $proceed$ to $False$ at \cref{second visit of end stop}.
	If the backbone is a path, then algorithm considers each vertex exactly once and sets the variable $proceed$ to $False$ at \cref{end of backbone stop}.
	Thus, \cref{from pseudocaterpillar to circular-arc graph} runs in linear time.
\end{proof}

Before proving \cref{CC and CP linear time on pseudocaterpillars}, we first need to define the notions of \emph{clique cover} and \emph{circular-arc graph}.

\begin{definition}[Clique cover]
	A \emph{clique cover} of a graph $H$ is a partition $\mathcal{Q}$ of $V(H)$ such that for each $Q_i \in \mathcal{Q}$, $Q_i$ is a clique.
	A \emph{minimum clique cover} is a clique cover containing a minimum number of cliques.
\end{definition}

\begin{definition}[Circular-arc graph]
	A \emph{circular-arc graph} $H$ is the intersection graph of a set of arcs on the circle, that is, each vertex of $H$ can be represented as an arc on the circle and there is an edge between two vertices if and only if the corresponding two arcs intersect.
\end{definition}

Finally, we prove the main result of this section.

\begin{theorem}\label{CC and CP linear time on pseudocaterpillars}
	{\CC} and {\CP} can be solved in linear time on coloured pseudocaterpillars.
\end{theorem}
\begin{proof}
	We note that checking whether a given graph is colourful can be done in linear time, for instance with a breadth-first search.

	Let $G$ be a coloured pseudocaterpillar with backbone $B$.
	First, we prove that a solution to {\CC} for $G$ can be found in linear time.
	We apply the preprocessing to $G$, as defined in \cref{preprocessing}, and denote by $S_p$ the set of edges that have been removed.
	Hence, $G-S_p$ contains colourful stars.
	If $G-S_p$ is colourful, then $S_p$ is an optimal solution to {\CC}.
	Otherwise, denote by $G'$ the connected component of $G-S_p$ containing the backbone $B$.
	If $G'$ contains exactly one colour that appears more than once, then according to \cref{exactly one colour appears at least twice} {\CC} and {\CP} can be solved in linear time.
	Therefore, we assume that $G'$ contains at least two colours that appear at least twice.
	Let $A$ be the multiset of ordered pairs obtained by \cref{from pseudocaterpillar to circular-arc graph} with input $G'$.
	According to \cref{algo is linear}, $A$ can be obtained in linear time.

	Following \cref{bad path equiv arc}, each ordered pair $(x,y)$ in $A$ corresponds to a critical bad path $P$ from $x$ to $y$ in~$G$.
	Using the algorithm of \citeauthor{MR1143909}~\cite{MR1143909}, we obtain a minimum clique cover $\mathcal{Q}$ of the circular-arc graph $H$ represented by $A$ in linear time.
	As mentioned in~\cite{MR1143909}, there are two types of cliques in a circular-arc graph.
	The first type of cliques contains three arcs which do not contain a common point of the circle.
	The second type of cliques contain a common point of the circle and are called \emph{linear cliques}.
	It is proved that in any circular-arc graph, there exists a minimum clique cover made of linear cliques only, unless the graph is a complete graph and the unique maximal clique is not linear.
	Hence, either all cliques in $\mathcal{Q}$ are linear or the unique clique in $\mathcal{Q}$ is not linear.
	In the following, we first deal with the case where $\mathcal{Q}$ contains linear cliques and then consider the case where the unique clique in $\mathcal{Q}$ is not linear.

	Suppose that all cliques in $\mathcal{Q}$ are linear.
	Let $S'$ be an empty set of edges.
	Choose a clique $Q_i \in \mathcal{Q}$.
	From our construction of $A$, and thus of $H$, each vertex $z \in Q_i$ corresponds to a critical bad path $P_z$ in~$G$.
	Let $D_i := \bigcap_{z \in Q_i} P_z$, and notice that, since $Q_i$ is linear, $|D_i \cap B| >0$.
	Then, choose an edge $e \in D_i \cap B$, and add $e$ to $S'$.
	We claim that, once each clique in $\mathcal{Q}$ has been processed, that is, when $|S'| = |\mathcal{Q}|$, the set $S'$ is an optimal solution to {\CC} on~$G$.
	Notice that $S'$ can be computed in linear time.
	As stated before, each critical bad path in $G'$ maps to an ordered pair in $A$, which corresponds to a vertex of $H$.
	Hence, a linear clique $Q_i$ of $H$ corresponds to a set of critical bad paths sharing a common subpath $D_i$.
	The set $S'$ contains an edge in $D_i \cap B$ for each $Q_i \in \mathcal{Q}$, and hence there is no critical bad path in $G'-S'$.
	Since $S' \subseteq B$ and each critical bad path has an edge in $S'$, then by \cref{critical bad paths are enough} $S'$ is a solution to {\CC} on $G'$.
	Moreover, since $\mathcal{Q}$ is minimum, $S'$ is an optimal solution for $G'$.
	Thus, the set $S := S_p \cup S'$ is an optimal solution to {\CC} on~$G$.

	Now, suppose that the unique clique in $\mathcal{Q}$ is not linear, which implies that $H$ is a complete graph.
	As mentioned before, there exist at least three arcs that do not have a common point on the circle.
	Let $z$ denote one of these arcs such that it does not strictly contain any other arc.
	Clearly, since $H$ is a complete graph, $z$ overlaps every other arc.
	Also, since $z$ does not strictly contain any other arc, if one extremity of $z$ does not overlap another arc $z'$, then its other extremity must overlap $z'$.
	Let $P_z$ be the colour-minimal bad path, with endpoints $u$ and $v$, corresponding to $z$.
	Denote by $u'$ and $v'$ the neighbours of $u$ and $v$ in $P_z$, respectively.
	Then, the set $S' := \{\{u,u'\} , \{v,v'\}\}$ is a solution to {\CC} on $G'$.
	Since the intersection of all the critical bad paths in $G'$ is empty, any solution to {\CC} on $G'$ contains at least $2$ edges.
	Hence, $S'$ is an optimal solution and $S := S_p \cup S'$ is an optimal solution to {\CC} on~$G$.

	Finally, we can detect each connected component of $G-S$ in linear time (for instance, with a breadth-first search).
	Thus, we can construct the partition $\pi$ of $V$ such that each part corresponds to a connected component of $G-S$ in linear time.
	Obviously, $\pi$ is a solution to {\CP} on~$G$.
	Since $S$ is optimal, the partition $\pi$ is optimal.
\end{proof}

Note that, for any two critical bad paths $P$ and $P'$, if $E(P') \subset E(P)$, then $P$ does not have to be represented.
While this observation can help decreasing the number of vertices in the circular-arc graph, it does not affect the overall worst-case linear-time complexity of the algorithm.

\begin{figure}[htb]
	\centering
	\begin{tabular}{C{0.45\textwidth}C{0.45\textwidth}}
		\begin{tikzpicture}[scale=.95]
	\tikzset{every node/.style={draw,circle,fill=gray,inner sep=2.5pt,outer sep=0},
		every label/.append style={rectangle},
		label distance=10pt};

	\node[fill=colorSeven] (c1) at (0,0) {};
	\node[fill=colorTwelve] (c2) at ($(c1) + (45:1)$) {};
	\node[fill=colorNine] (c3) at ($(c2) + (10:1)$) {};
	\node[fill=colorSeven] (c4) at ($(c3) + (-10:1)$) {};
	\node[fill=colorOne] (c5) at ($(c4) + (-45:1)$) {};
	\node[fill=colorFive] (c6) at ($(c5) + (-135:1.25)$) {};
	\node[fill=colorNine] (c7) at ($(c6) + (-180:1.5)$) {};

	\draw (c1) -- (c2);
	\draw (c2) -- (c3);
	\draw (c3) [dashed] -- (c4);
	\draw (c4) -- (c5);
	\draw (c5) [dashed] -- (c6);
	\draw (c6) [dashed] -- (c7);
	\draw (c7) [dashed] -- (c1);

	\node[fill=colorFour] (d1) at ($(c1) + (160:1)$) {};
    \node[fill=colorFour] (e1) at ($(c1) + (200:1)$) {};
    \draw (c1) [dotted] -- (d1);
    \draw (c1) -- (e1);

    \node[fill=colorOne] (d2) at ($(c2) + (115.5:1)$) {};
    \draw (c2) -- (d2);

    \node[fill=colorFive] (d4) at ($(c4) + (97.5:1)$) {};
    \node[fill=colorFour] (e4) at ($(c4) + (67.5:1)$) {};
    \node[fill=colorSeven] (f4) at ($(c4) + (37.5:1)$) {};
    \draw (c4) -- (d4);
    \draw (c4) -- (e4);
    \draw[dotted] (c4) -- (f4);

    \node[fill=colorNine] (d5) at ($(c5) + (0:1)$) {};
    \draw (c5) -- (d5);

    \node[fill=colorTwelve] (d6) at ($(c6) + (-47.5:1)$) {};
    \node[fill=colorSeven] (e6) at ($(c6) + (-87.5:1)$) {};
    \draw (c6) -- (d6);
    \draw (c6) -- (e6);

    \node[fill=colorOne] (d7) at ($(c7) + (-180+22.5:1)$) {};
    \node[fill=colorTwelve] (e7) at ($(c7) + (-180+52.5:1)$) {};
    \node[fill=colorFour] (f7) at ($(c7) + (-180+82.5:1)$) {};
    \node[fill=colorOne] (g7) at ($(c7) + (-180+112.5:1)$) {};
    \draw[dotted] (c7) -- (d7);
    \draw (c7) -- (e7);
    \draw (c7) -- (f7);
    \draw (c7) -- (g7);

\end{tikzpicture}
  &
		\begin{tikzpicture}[scale=1]
\def\centerarc[#1](#2)(#3:#4:#5){\draw[#1] ($(#2)+({#5*cos(#3)},{#5*sin(#3)})$) arc (#3:#4:#5);}

\centerarc[colorFour, very thick](0,0)(48:177:1.4);
    \centerarc[colorFour, very thick](0,0)(183:237:1.4);
    \centerarc[colorFour, very thick](0,0)(-123:42:1.4);

    \centerarc[colorFive, very thick](0,0)(48:297:1);
    \centerarc[colorFive, very thick](0,0)(-57:42:1);

    \centerarc[colorSeven, very thick](0,0)(48:177:1.2);
    \centerarc[colorSeven, very thick](0,0)(183:297:1.2);
    \centerarc[colorSeven, very thick](0,0)(-57:42:1.2);

    \centerarc[colorTwelve, very thick](0,0)(138:237:0.8);
    \centerarc[colorTwelve, very thick](0,0)(243:297:0.8);
    \centerarc[colorTwelve, very thick](0,0)(-57:132:0.8);

    \centerarc[colorNine, very thick](0,0)(93:237:1.6);
    \centerarc[colorNine, very thick](0,0)(243:357:1.6);
    \centerarc[colorNine, very thick](0,0)(3:87:1.6);

    \centerarc[colorOne, very thick](0,0)(3:132:1.8);
    \centerarc[colorOne, very thick](0,0)(138:237:1.8);
    \centerarc[colorOne, very thick](0,0)(-117:-3:1.8);

    \draw [dashed, very thick] ($(0,0)+(67.5:0.5)$) to ($(0,0)+(67.5:2)$);
    \draw [dashed, very thick] ($(0,0)+(210:0.5)$) to ($(0,0)+(210:2)$);
    \draw [dashed, very thick] ($(0,0)+(270:0.5)$) to ($(0,0)+(270:2)$);
    \draw [dashed, very thick] ($(0,0)+(330:0.5)$) to ($(0,0)+(330:2)$);
\end{tikzpicture}    \\
	\end{tabular}
	\caption{On the left, a pseudocaterpillar $G$: dotted edges are removed in the preprocessing; dashed edges are obtained from \cref{from pseudocaterpillar to circular-arc graph}; dotted and dashed edges form an optimal solution to {\CC}.
		On the right, an arc representation of the circular-arc graph constructed from the critical bad paths in $G$ (after preprocessing): dashed segments represent a minimum clique cover and correspond to the dashed edges in~$G$.}
	\label{1-caterpillar full resolution}
\end{figure}

\Cref{1-caterpillar full resolution} gives an example of a pseudocaterpillar $G$ and its representation as circular-arc graph $H$ where a minimum clique cover of $H$ represents an optimal solution to {\CC} on~$G$.

\medskip
\noindent\textbf{Generalisation of the result.}

Now, we slightly generalise our results on pseudocaterpillars to what we call \emph{necklace graphs with colourful beads}.
We say that a graph is a \emph{necklace graph} if it can be obtained by adding edges between connected components, called \emph{beads}, such that the newly added edges induce a path or a cycle, called the \emph{backbone}, and exactly one vertex of each bead is contained in the backbone.
Note that necklace graphs are, by definition, connected.
We say that a graph is a necklace graph with colourful beads if every bead is colourful.
The general idea is to transform the input necklace graph with colourful beads into a pseudocaterpillar with colourful stars, where beads will be represented as leaves, and then use our earlier results on pseudocaterpillars from to obtain a linear time algorithm for {\CC} and {\CP}.

\begin{figure}[htb]
    \centering
    \begin{tikzpicture}[scale=0.95]
\usetikzlibrary{shapes}
	\tikzset{every node/.style={draw,circle,fill=gray,inner sep=2.5pt,outer sep=0},
		every label/.append style={rectangle},
		label distance=10pt};

    \tikzset{bead/.style={draw,fill=white,regular polygon, regular polygon sides=3}}

	\node[fill=colorSeven] (c1) at (0,0) {};
	\node[fill=colorTwelve] (c2) at ($(c1) + (45:1)$) {};
	\node[fill=colorNine] (c3) at ($(c2) + (10:1)$) {};
	\node[fill=colorSeven] (c4) at ($(c3) + (-10:1)$) {};
	\node[fill=colorEight] (c5) at ($(c4) + (-45:1)$) {};
	\node[fill=colorFive] (c6) at ($(c5) + (-135:1.25)$) {};
	\node[fill=colorNine] (c7) at ($(c6) + (-180:1.5)$) {};

	\draw[] (c1) -- (c2);
	\draw[] (c2) -- (c3);
	\draw[] (c3) [] -- (c4);
	\draw[] (c4) -- (c5);
	\draw[] (c5) [] -- (c6);
	\draw[] (c6) [] -- (c7);
	\draw[] (c7) [] -- (c1);

\draw[dashed,rotate=0] (c1)++(-.65,0) ellipse (1.2cm and .5cm);
    \draw[dashed,rotate=-70] (c2)++(-.65,0) ellipse (1.2cm and .5cm);
    \draw[dashed,rotate=-90] (c3)++(-.65,0) ellipse (1.2cm and .5cm);
    \draw[dashed,rotate=-110] (c4)++(-.65,0) ellipse (1.2cm and .5cm);
	\draw[dashed,rotate=180] (c5)++(-.65,0) ellipse (1.2cm and .5cm);
	\draw[dashed,rotate=135] (c6)++(-.65,0) ellipse (1.2cm and .5cm);
	\draw[dashed,rotate=45] (c7)++(-.65,0) ellipse (1.2cm and .5cm);
\end{tikzpicture}
     \caption{A schematic representation of a necklace graph with colourful beads. The backbone is represented with plain edges and the beads with dashed ellipses.}
    \label{fig:necklace graph}
\end{figure}

First, we prove a similar result as \cref{exactly one colour appears at least twice} in the context of necklace graphs.

\begin{lemma}\label{at most one colour appears at least twice necklace}
{\CC} and {\CP} can be solved in linear time on coloured necklace graphs with colourful beads where at most one colour appears at least twice.
\end{lemma}
\begin{proof}
	Let $G$ be a coloured necklace graph with colourful beads and backbone $B$.
	If no colour appears more than once, then $G$ is colourful and it can be checked in linear time.

	Suppose that there exists a unique colour $\gamma$ that appears at least twice in~$G$.
	Fix a vertex $u$ of colour $\gamma$ and let $X$ be the set of vertices in the same bead as~$u$.
	Let $v$ be the vertex in $V(B) \cap X$ (note that $u = v$ may hold) and consider a minimal solution $S$ to {\CC} on~$G$.
	Observe that if $u=v$, that is, if $u$ is on the backbone, then $S$ does not contain any edge in $G[X]$ (or else $S$ is not minimum).
	If the size of a minimum $u{,}v$-cut (a minimum set of edges intersecting every path between $u$ and $v$ in $G$) is at least $2$, then it is always possible to find $S$ such that $S$ does not contain any edge from $G[X]$, for instance by removing the two edges in $B$ incident to $v$ instead of the edges in the $u{,}v$-cut.
	So in this case, we can assume that $u \in V(B)$.
	On the other hand, if the $u{,}v$-cut is of size exactly one, then it is equivalent as if $v$ were the only neighbour of $u$ (the edge $\{u,v\}$ being a minimum cut of size one).
	It is now easy to see the parallel with colourful pseudocaterpillars where exactly one colour appears at least twice, and thus we can use a simple reduction and obtain a solution to {\CC} and {\CP} for $G$ in linear time using \cref{exactly one colour appears at least twice}.
\end{proof}

Now, we can focus on the case where the input graph contains at least two colours appearing at least twice.
The main idea consists in finding a reduction from {\CC} on coloured necklace graphs with colourful beads to {\CC} on coloured pseudocaterpillars with colourful stars.

The following lemma can be proved similarly as \cref{at least two colours appear twice}, using the fact that all beads are colourful and generalising the notation $C[u]$, for some vertex $u$ in the backbone, to be the set of colours of the vertices of the colourful bead containing $u$.

\begin{lemma}\label{at least two colours appear twice necklace}
	Let $G$ be a coloured necklace graph with colourful beads such that at least two colours appear at least twice in~$G$.
	Then there exists an optimal solution $S$ to {\CC} for $G$ such that $S \subseteq B$, where $B$ is the backbone of~$G$.
\end{lemma}

\begin{theorem}\label{necklace to pseudocaterpillar in linear time}
	{\CC} and {\CP} can be solved in linear time on coloured necklace graphs with colourful beads.
\end{theorem}
\begin{proof}
	Let $G$ be a coloured necklace graph with colourful beads with backbone~$B$.
	We assume that $G$ contains at least two colours that appear at least twice, otherwise we use \cref{at most one colour appears at least twice necklace} to obtain a solution in linear time.
	We create a coloured graph $G'$ with vertex set $V(G') := V(G)$ (the vertices keep their original colour) and backbone $B' := B$.
	Then, for each $u \in V(B')$, we connect $u$ to all the vertices $v$ such that in $G$ the vertices $u$ and $v$ belong to the same bead.
	Since the beads are colourful, we obtain that $G'$ is a coloured pseudocaterpillar with colourful stars.
	By \cref{at least two colours appear twice}, there exists an optimal solution $S_{G'}$ to {\CC} on $G'$ such that $S_{G'} \subseteq B'$.
	It is clear that $S_{G'}$ is also a solution for~$G$.
	Moreover, following \cref{at least two colours appear twice necklace}, we know that there exists an optimal solution $S_G$ for $G$ such that $S_G \subseteq B$.
	Suppose that $|S_G| < |S_{G'}|$ and observe that $S_G$ is also a solution for $G'$.
	A contradiction with $S_{G'}$ being optimal.

	Also, note that an optimal solution to {\CP} can be obtained in linear time from $S_{G'}$ (since $S_{G'} \subseteq B'$).
	According to \cref{CC and CP linear time on pseudocaterpillars}, $S_{G'}$ can be obtained in linear time.
\end{proof}
 
\subsection{Complexity dichotomies}\label{dichotomies}

Our results from \cref{hard-caterpillars,easy-caterpillars} imply two complexity dichotomies on trees: one with respect to the maximum degree and the other with respect to the smallest value $k$ such that the input tree is a $k$\=/caterpillar (and thus not a $(k-1)$-caterpillar).

\dichotomytreeboundeddegree*
\begin{proof}
    If the input graph is a path, then \cref{CC and CP linear time on pseudocaterpillars} implies that both {\CC} and {\CP} are linear-time solvable.
    On the other hand, if the input graph is a tree that can contain vertices with degree $3$ or more, then by \cref{np-hard caterpillars} both problems are \NP-complete.
\end{proof}

\dichotomykcaterpillar*
\begin{proof}
    If $k \leq 1$, then \cref{CC and CP linear time on pseudocaterpillars} implies that both problems can be solved in linear-time.
    However, if $k \geq 2$, then \cref{np-hard caterpillars} shows that the problem is \NP-complete.
\end{proof}

Note that these results naturally extend to the case when the input graph can be disconnected by considering each connected component separately.
In particular, both problems are linear-time solvable on forest if every connected component is a $1$-caterpillar (which includes the case when the component is a path) and \NP-complete otherwise. 
\section{Colourful Components on planar graphs}\label{subcubic graphs}

In \cite{MR2988970}, the authors prove that {\CC} is \NP-complete even when restricted to $3$\=/coloured graphs with maximum degree $6$.
Using a similar reduction from \textsc{Planar $3$\=/SAT}, we show how the vertices of degree $6$ can be replaced with gadgets only containing vertices of degree $4$, or $3$, if we relax the number of colours from $3$ to $5$, or to $12$, respectively.

\problemdef{\textsc{Planar $3$\=/SAT}}{A $3$\=/\textsc{CNF} formula $\phi$ in which the bipartite graph of variables and clauses is planar.}{Is there a satisfying assignment for $\phi$?}

Note that \textsc{Planar $3$\=/SAT} is \NP-complete \cite{MR652906} and it can be shown that it remains so even if each clause contains exactly $3$ literals \cite{MR3962621}.
In the following, we consider the latter version.

\begin{construction}\label{from phi to graph}
	Given an instance $\phi$ of \textsc{Planar $3$\=/SAT}, that is a set of $m$ clauses $C_1,C_2,\dots,C_m$ on $n$ variables, we construct the graph $G$ such that:
	\begin{itemize}
		\item For each variable $x$ in $\phi$, let $m_x$ denote the number of clauses in which $x$ appears.
    We construct a \emph{variable cycle} of length $4m_x$ in $G$ with vertices $V_{x} := \{x_j^1, x_j^2, x_j^3, x_j^4 ~|~ x \in C_j\}$ with an arbitrary fixed cyclic ordering of the clauses containing $x$.
    The vertices are coloured alternatively with two colours $c_o$ and $c_e$ such that $c(x_j^1) = c(x_j^3)= c_o$ and $c(x_j^2) = c(x_j^4)= c_e$, for all $j$ such that $x\in C_j$.
    For simplicity, edges $\{x_j^1,x_j^2\}$ and $\{x_j^3,x_j^4\}$ are called \emph{odd edges}, and the other edges in the variable cycle are called \emph{even edges}.

		\item For each clause $C_j$ containing three variables $p$, $q$ and $r$, we construct a \emph{clause gadget}.
      We propose two types of gadgets (see \cref{gadgets for max degree 4 and 3} for an example):

	      (i) The gadget $\mathcal{A}_j^4$ is made of a cycle of length $3$, with vertices $a_j^1$, $a_j^2$ and $a_j^3$ such that each $a_j^i$ is given colour $i$, different from $c_o$ and $c_e$.
        We define how the vertices from $V_p$ are connected to $\mathcal{A}_j^4$.
        If the variable $p$ appears as a positive literal in $C_j$, connect the vertices $p_j^1$ to $a_j^1$ and $p_j^2$ to $a_j^2$.
        Otherwise, if $p$ occurs as a negative literal, connect the vertices $p_j^2$ to $a_j^1$ and $p_j^3$ to $a_j^2$.
        Do the same for the variables $q$ and $r$ by connecting the corresponding vertices in $V_{q}$ to $a_j^2$ and $a_j^3$, and the corresponding vertices in $V_{r}$ to $a_j^3$ and $a_j^1$.
        Notice that the vertices in $\mathcal{A}_j^4$ have degree  $4$.

          (ii) The gadget $\mathcal{A}_j^3$ is made of a cycle of length $9$ with vertices labelled $a_j^1,\dots,a_j^9$ and an additional vertex $a_j^{10}$ connected to $a_j^2$, $a_j^5$ and $a_j^8$.
        We set the colour $i$ to each vertex $a_j^i$, different from $c_o$ and $c_e$.
        We define how the vertices from $V_p$ are connected to $\mathcal{A}_j^3$.
        If the variable $p$ appears as a positive literal in $C_j$, connect the vertices $p_j^1$ to $a_j^1$ and $p_j^2$ to $a_j^3$.
        Otherwise, if $p$ occurs as a negative literal, connect the vertices $p_j^2$ to $a_j^1$ and $p_j^3$ to $a_j^3$.
        Do the same for the variables $q$ and $r$ by connecting the corresponding vertices in $V_q$ to $a_j^4$ and $a_j^6$, and the corresponding vertices in $V_r$ to $a_j^7$ and $a_j^9$.
        Notice that the vertices in $\mathcal{A}_j^3$ have degree  $3$.
	\end{itemize}
	Since the bipartite graph of variables and clauses of $\phi$ is planar and each vertex can be replaced by a clause or vertex gadget, with a correct cyclic ordering of the clauses for each variable, the resulting graph $G$ is planar.
\end{construction}

Note that \cref{from phi to graph} can be done in polynomial time.

\begin{figure}[tb]
	\centering
	\begin{tikzpicture}[scale=0.7]
	\tikzset{every node/.style={draw,circle,fill=gray,minimum size=5pt,inner sep=2.5pt,outer sep=0},
		every label/.append style={rectangle},
		label distance=-1pt};

	\begin{scope} \coordinate (p) at (150:5);
		\foreach \a/\t in {45/p-1, 30/p0, 10/p1, -10/p2, -30/p3, -45/p4} {
				\coordinate (\t) at ([shift=(-\a-30:2.5)]p) {};
			}
		\node (p0) at (p0) {};
		\node[label=175:$\mathstrut p_j^1$,fill=white] (p1) at (p1) {};
		\node[label=175:$\mathstrut p_j^2$] (p2) at (p2) {};
		\node[fill=white] (p3) at (p3) {};

		\draw[] (p0) -- (p1) -- (p2) -- (p3);
		\draw[dashed] (p-1) -- (p0);
		\draw[dashed] (p3) -- (p4);

		\coordinate (q) at (30:5);
		\foreach \a/\t in {45/q-1, 30/q0, 10/q1, -10/q2, -30/q3, -45/q4} {
				\coordinate (\t) at ([shift=(-\a-150:2.5)]q) {};
			}
		\node[fill=white] (q0) at (q0) {};
		\node[label=5:$\mathstrut q_j^2$] (q1) at (q1) {};
		\node[label=5:$\mathstrut q_j^3$,fill=white] (q2) at (q2) {};
		\node (q3) at (q3) {};

		\draw[] (q0) -- (q1) -- (q2) -- (q3);
		\draw[dashed] (q-1) -- (q0);
		\draw[dashed] (q3) -- (q4);

		\coordinate (r) at (-90:5);
		\foreach \a/\t in {45/r-1, 30/r0, 10/r1, -10/r2, -30/r3, -45/r4} {
				\coordinate (\t) at ([shift=(-\a+90:2.5)]r) {};
			}
		\node (r0) at (r0) {};
		\node[label=below:$\mathstrut r_j^1$,fill=white] (r1) at (r1) {};
		\node[label=below:$\mathstrut r_j^2$] (r2) at (r2) {};
		\node[fill=white] (r3) at (r3) {};

		\draw[] (r0) -- (r1) -- (r2) -- (r3);
		\draw[dashed] (r-1) -- (r0);
		\draw[dashed] (r3) -- (r4);

{
\foreach \a/\t in {1/aj1, 2/aj2, 3/aj3} {
				\coordinate (\t) at (-30+-\a*360/3:6/5);
			}

		\node[fill=colorEight,,label=-120:$\mathstrut a_j^1$] (aj1) at (aj1) {};
		\node[fill=colorNine,label=$\mathstrut a_j^2$] (aj2) at (aj2) {};
		\node[fill=colorEleven,label=-60:$\mathstrut a_j^3$] (aj3) at (aj3) {};

		\draw[] (aj1) -- (aj2) -- (aj3) --(aj1);
		}
\draw[] (p1) -- (aj1);
		\draw[] (p2) -- (aj2);
		\draw[] (q1) -- (aj2);
		\draw[] (q2) -- (aj3);
		\draw[] (r1) -- (aj3);
		\draw[] (r2) -- (aj1);

\node[draw=none,fill=none,rectangle] (caption1) at (0,-4.7) {Gadget with vertices of degree $4$.};

	\end{scope}

	\begin{scope}[xshift=9.5cm] \coordinate (p) at (150:5);
		\foreach \a/\t in {45/p-1, 30/p0, 10/p1, -10/p2, -30/p3, -45/p4} {
				\coordinate (\t) at ([shift=(-\a-30:2.5)]p) {};
			}
		\node (p0) at (p0) {};
		\node[label=175:$\mathstrut p_j^1$,fill=white] (p1) at (p1) {};
		\node[label=175:$\mathstrut p_j^2$] (p2) at (p2) {};
		\node[fill=white] (p3) at (p3) {};

		\draw[] (p0) -- (p1) -- (p2) -- (p3);
		\draw[dashed] (p-1) -- (p0);
		\draw[dashed] (p3) -- (p4);

		\coordinate (q) at (30:5);
		\foreach \a/\t in {45/q-1, 30/q0, 10/q1, -10/q2, -30/q3, -45/q4} {
				\coordinate (\t) at ([shift=(-\a-150:2.5)]q) {};
			}
		\node[fill=white] (q0) at (q0) {};
		\node[label=5:$\mathstrut q_j^2$] (q1) at (q1) {};
		\node[label=5:$\mathstrut q_j^3$,fill=white] (q2) at (q2) {};
		\node (q3) at (q3) {};

		\draw[] (q0) -- (q1) -- (q2) -- (q3);
		\draw[dashed] (q-1) -- (q0);
		\draw[dashed] (q3) -- (q4);

		\coordinate (r) at (-90:5);
		\foreach \a/\t in {45/r-1, 30/r0, 10/r1, -10/r2, -30/r3, -45/r4} {
				\coordinate (\t) at ([shift=(-\a+90:2.5)]r) {};
			}
		\node (r0) at (r0) {};
		\node[label=below:$\mathstrut r_j^1$,fill=white] (r1) at (r1) {};
		\node[label=below:$\mathstrut r_j^2$] (r2) at (r2) {};
		\node[fill=white] (r3) at (r3) {};

		\draw[] (r0) -- (r1) -- (r2) -- (r3);
		\draw[dashed] (r-1) -- (r0);
		\draw[dashed] (r3) -- (r4);

		{
		\tikzset{every node/.append style={fill=black}}

		\foreach \a/\t in {1/aj1, 2/aj3, 3/aj4, 4/aj6, 5/aj7, 6/aj9} {
				\coordinate (\t) at (-120-\a*360/6:8/5);
			}

		\node[fill=colorFour,label=-120:$\mathstrut a_j^1$] (aj1) at (aj1) {};
		\node[fill=colorTwelve,label=60:$\mathstrut a_j^3$] (aj3) at (aj3) {};
		\node[fill=colorEight,label={[label distance=-5pt]150:$\mathstrut a_j^2$}] (aj2) at ($(aj1)!0.5!(aj3)$) {};

		\node[fill=colorSix,label=120:$\mathstrut a_j^4$] (aj4) at (aj4) {};
		\node[fill=colorSeven,label=-60:$\mathstrut a_j^6$] (aj6) at (aj6) {};
		\node[fill=colorNine,label={[label distance=-5pt]30:$\mathstrut a_j^5$}] (aj5) at ($(aj4)!0.5!(aj6)$) {};

		\node[fill=colorTen,label=0:$\mathstrut a_j^7$] (aj7) at (aj7) {};
		\node[fill=colorFive,label=180:$\mathstrut a_j^9$] (aj9) at (aj9) {};
		\node[fill=colorEleven,label=-90:$\mathstrut a_j^8$] (aj8) at ($(aj7)!0.5!(aj9)$) {};

		\node[fill=colorOne,label=90:$\mathstrut a_j^{10}$] (aj10) at (0,0) {};

		\draw[] (aj1) -- (aj2) -- (aj3) -- (aj4) -- (aj5) -- (aj6) -- (aj7) -- (aj8) -- (aj9) -- (aj1);
		\draw[] (aj10) --(aj2) (aj10) -- (aj5) (aj10) -- (aj8);
		}
\draw[] (p1) -- (aj1);
		\draw[] (p2) -- (aj3);
		\draw[] (q1) -- (aj4);
		\draw[] (q2) -- (aj6);
		\draw[] (r1) -- (aj7);
		\draw[] (r2) -- (aj9);

\node[draw=none,fill=none,rectangle] (caption1) at (0,-4.7) {Gadget with vertices of degree $3$.};

	\end{scope}

\end{tikzpicture}
 	\caption{Two clause gadgets $\mathcal{A}_j^4$ (left) and $\mathcal{A}_j^3$ (right) of a clause $C_j := (p \vee \bar{q} \vee r)$. White vertices have colour $c_o$, grey vertices have colour $c_e$.}\label{gadgets for max degree 4 and 3}
\end{figure}

\begin{theorem}
	{\CC} is \NP-complete on $5$\=/coloured planar graphs with maximum degree $4$ and on $12$-coloured planar graphs with maximum degree $3$.
\end{theorem}
\begin{proof}
	Let $\phi$ be an instance of \textsc{Planar $3$\=/SAT} with $m$ clauses and $n$ variables and $G$ be the graph obtained through \cref{from phi to graph}, using either $\mathcal{A}_j^4$ gadgets or $\mathcal{A}_j^3$ gadgets.
	Notice that, depending on the type of gadget used, $G$ is either a $5$\=/coloured graph with maximum degree $4$ or a $12$-coloured graph with maximum degree $3$.
	We claim that there exists a satisfying assignment $\beta$ for $\phi$ if and only if there exists a solution $S$ to {\CC} such that $|S| = 10m$.

	Let $\beta$ be a satisfying assignment for $\phi$.
	If $x = False$ in $\beta$, then remove the $2m_x$ odd edges in its variable cycle (of length $4m_x$), that is, edges of the type $\{x_j^1,x_j^2\}$ and $\{x_j^3,x_j^4\}$, for every clause $C_j$ containing $x$.
    Similarly, if $x = True$ in $\beta$, then remove the $2m_x$ even edges in its variable cycle.
	Since every clause contains exactly $3$ literals, $\sum_{1 \leq i \leq n} 4m_i / 2 = 4 \times 3m/2 = 6m$ edges have been removed and $G$ does not have any bad paths that contain only vertices from the variable cycles.
	Now, fix a clause $C_j$ and denote its variables by $p$, $q$ and~$r$.
	Without loss of generality assume that $p$ satisfies $C_j$.
	Remove the $4$ edges between the clause gadget of $C_j$ and the vertices in $V_q$ and $V_r$.
	Without loss of generality, assume that $C_j$ contains $p$ as a positive literal.
	Hence the gadget is connected to the vertices $p_j^1$ and $p_j^2$.
	Since $p$ is set to true in $\beta$, the even edges of the variable cycle have been removed.
	Therefore, the vertices $p_j^1$, $p_j^2$, together with the vertices in the clause gadget form a colourful component.
	The other case where $C_j$ contains $p$ as a negative literal is similar, and the vertices $p_j^2$, $p_j^3$ and those in the clause gadget form a colourful component.
	If the clause gadget is $\mathcal{A}_j^4$, then the colourful component is of size $5$.
	If the clause gadget is $\mathcal{A}_j^3$, then the colourful component is of size $12$.
	Since $4$ edges are removed for each clause, a total of $4m$ edges between variable cycles and clause gadgets are removed.
	Let $S$ be the set of removed edges, and notice that $S$ is a solution to {\CC} such that $|S| = 6m + 4m = 10m$.

    Let $S$ be a solution to {\CC} for $G$ such that $|S| = 10m$.
	Since the vertices of the variable cycles are coloured alternatively with $c_o$ and $c_e$, for each variable cycle, $S$ contains all its odd edges or all its even edges.
    Let $S' \subseteq S$ be a set of pairwise non-adjacent edges in $S$ with maximum size whose edges have both endpoints in a variable cycle.
    In particular, notice that if a variable cycle has all its odd edges in $S$ but not all its even edges, then all its odd edges belong to $S'$, and similarly if all its even edges belong to $S$;
    if all edges of the variable cycle belong to $S$, that is, both odd and even edges, then $S'$ contains either all its odd edges or all its even edges.
    Hence, $S'$ contains at least $\sum_{1 \leq i \leq n} 4m_i / 2 = 6m$ edges.
    For a fixed clause $C_j$, we denote by $H_j$ the connected component containing the clause gadget of $C_j$ in $G-S'$.
    Note that no two clause gadgets belong to a same connected component of $G-S'$, and so for any two clauses $C_k$ and $C_\ell$ we have $E(H_k) \cap E(H_\ell) = \emptyset$.
    Let $F_j = S \cap E(H_j)$.
    Observe that, independently on whether the odd or even edges of the variable cycles connected to the clause gadget of $C_j$ belong to $S'$, if $|F_j| \leq 3$, then $H_j - F_j$ contains two vertices with the same colour $c_o$ or $c_e$ (from two variable cycles), and thus $H_j - F_j$ is not colourful; a contradiction with the fact that $G-S$ is colourful.
    Hence, $|F_j| \geq 4$.
    Since $|S| = 10m$ and $|S'| \geq 6m$, we get that $|S'| = 6m$ and that $|F_k| = 4$ for each clause $C_k$.
    Note that this also implies that, for each variable cycle, $S'$ contains exactly either all of its odd edges or all of its even edges.
    Also, notice that the $4$ edges in $S \cap F_k$ must be between the clause gadget and exactly two of the variable cycles attached to it; otherwise, at least two vertices with the same colour $c_o$ or $c_e$ (from two variable cycles) would belong to a same connected component of $G-S$.
    This implies that each clause gadget is adjacent to vertices of exactly one variable cycle in $G-S$.
	Let $\beta$ be an assignment for $\phi$ described as follows.
	For each variable $x$, if $S'$ contains the odd edges of the variable cycle of $x$, then set $x := False$ in $\beta$.
	Otherwise, if $S'$ contains the even edges, then set $x := True$ in $\beta$.
    Recall that each clause gadget is adjacent to vertices of exactly one variable cycle in $G-S$.
	To prove that $\beta$ is a satisfying assignment for $\phi$, we claim that if vertices of the variable cycle of a variable $x$ is connected to the clause gadget of a clause $C_j$ in $G-S$, then $C_j$ is satisfied by $x$.
    Suppose the contrary and assume, without loss of generality, that $x$ appears as a positive literal in $\phi$ but $x = False$ in $\beta$.
	This means that the odd edge $\{x_j^1,x_j^2\}$ belongs to $S'$, and hence to $S$, and that the even edge $\{x_j^2,x_j^3\}$ belongs to $G-S$.
	Therefore, $x_j^1$ and $x_j^3$ belong to a same colourful component of $G-S$.
	However, $c(x_j^1) = c(x_j^3)$, a contradiction.
	We conclude that $\beta$ is a satisfying assignment for $\phi$.
\end{proof}
 
\section{Conclusion}

We provide two complexity dichotomies on trees: one with respect to the maximum degree (see \cref{dichotomy tree bounded degree}) and the other with respect to the smallest value $k$ such that the input tree is a $k$\=/caterpillar (see \cref{dichotomy k-caterpillar}).
In order to complete the dichotomy on $k$\=/caterpillars with respect to the maximum degree, for every $k \in \mathbb{N}$, we ask the following question.
\begin{question}
    What is the complexity of {\CC} and {\CP} on $3$\=/caterpillars with maximum degree~$3$, $2$\=/caterpillars with maximum degree~$3$ and $2$\=/caterpillars with maximum degree~$4$?
\end{question}

We also prove that {\CC} remains \NP\=/complete on $5$\=/coloured planar graphs with maximum degree~$4$ and on $12$-coloured planar graphs with maximum degree~$3$.
A natural question is to ask whether this result is tight with respect to the number of colours and the maximum degree.
\begin{question}
    What is the complexity of {\CC} on $4$\=/coloured (planar) graphs with maximum degree~$4$ and on $11$\=/coloured (planar) graphs with maximum degree~$3$?
\end{question}

\paragraph*{Acknowledgements}
The authors wish to thank Marthe Bonamy for suggesting these problems and Paul Ouvrard for valuable discussions.

\printbibliography

@article {MR3392498,
    AUTHOR = {Adamaszek, Anna and Popa, Alexandru},
     TITLE = {Algorithmic and hardness results for the colorful components
              problems},
   JOURNAL = {Algorithmica},
  FJOURNAL = {Algorithmica. An International Journal in Computer Science},
    VOLUME = {73},
      YEAR = {2015},
    NUMBER = {2},
     PAGES = {371--388},
      ISSN = {0178-4617},
   MRCLASS = {05C85 (05C15 05C40 68Q25)},
  MRNUMBER = {3392498},
MRREVIEWER = {Louxin Zhang},
       DOI = {10.1007/s00453-014-9926-0},
       IGNOREURL= {https://doi.org/10.1007/s00453-014-9926-0},
}

@article {MR2323384,
    AUTHOR = {Avidor, Adi and Langberg, Michael},
     TITLE = {The multi-multiway cut problem},
   JOURNAL = {Theoret. Comput. Sci.},
  FJOURNAL = {Theoretical Computer Science},
    VOLUME = {377},
      YEAR = {2007},
    NUMBER = {1-3},
     PAGES = {35--42},
      ISSN = {0304-3975},
   MRCLASS = {90C35 (68W25 90C08)},
  MRNUMBER = {2323384},
MRREVIEWER = {Margarida P. Mello},
       DOI = {10.1016/j.tcs.2007.02.026},
       IGNOREURL= {https://doi.org/10.1016/j.tcs.2007.02.026},
}

@article {MR3759884,
    AUTHOR = {Bousquet, Nicolas and Daligault, Jean and Thomass\'{e}, St\'{e}phan},
     TITLE = {Multicut is {FPT}},
   JOURNAL = {SIAM J. Comput.},
  FJOURNAL = {SIAM Journal on Computing},
    VOLUME = {47},
      YEAR = {2018},
    NUMBER = {1},
     PAGES = {166--207},
      ISSN = {0097-5397},
   MRCLASS = {05C85 (05C40 68Q25 68R10)},
  MRNUMBER = {3759884},
MRREVIEWER = {Anita Pal},
       DOI = {10.1137/140961808},
       IGNOREURL= {https://doi.org/10.1137/140961808},
}

@incollection {MR2988970,
    AUTHOR = {Bruckner, Sharon and H\"{u}ffner, Falk and Komusiewicz, Christian
              and Niedermeier, Rolf and Thiel, Sven and Uhlmann, Johannes},
     TITLE = {Partitioning into colorful components by minimum edge
              deletions},
 BOOKTITLE = {Combinatorial pattern matching},
    SERIES = {Lecture Notes in Comput. Sci.},
    VOLUME = {7354},
     PAGES = {56--69},
 PUBLISHER = {Springer, Heidelberg},
      YEAR = {2012},
   MRCLASS = {68Q17 (05C70 68Q25 68R10)},
  MRNUMBER = {2988970},
       DOI = {10.1007/978-3-642-31265-6_5},
       IGNOREURL= {https://doi.org/10.1007/978-3-642-31265-6_5},
}

@incollection {MR3968574,
    AUTHOR = {Bulteau, Laurent and Dabrowski, Konrad K. and Fertin,
              Guillaume and Johnson, Matthew and Paulusma, Dani\"{e}l and
              Vialette, St\'{e}phane},
     TITLE = {Finding a small number of colourful components},
 BOOKTITLE = {30th {A}nnual {S}ymposium on {C}ombinatorial {P}attern
              {M}atching},
    SERIES = {LIPIcs. Leibniz Int. Proc. Inform.},
    VOLUME = {128},
     PAGES = {Art. No. 20, 14},
 PUBLISHER = {Schloss Dagstuhl. Leibniz-Zent. Inform., Wadern},
      YEAR = {2019},
   MRCLASS = {68R10 (68Q17 68Q25)},
  MRNUMBER = {3968574},
}

@incollection {MR3991231,
    AUTHOR = {Chleb\'{i}kov\'{a}, Janka and Dallard, Cl\'{e}ment},
     TITLE = {Towards a complexity dichotomy for colourful components
              problems on {$k$}-caterpillars and small-degree planar graphs},
 BOOKTITLE = {Combinatorial algorithms},
    SERIES = {Lecture Notes in Comput. Sci.},
    VOLUME = {11638},
     PAGES = {136--147},
 PUBLISHER = {Springer, Cham},
      YEAR = {2019},
   MRCLASS = {68R10 (68Q17 68Q25)},
  MRNUMBER = {3991231},
       DOI = {10.1007/978-3-030-25005-8_12},
       IGNOREURL= {https://doi.org/10.1007/978-3-030-25005-8_12},
}

@book {MR1637890,
    AUTHOR = {Papadimitriou, Christos H. and Steiglitz, Kenneth},
     TITLE = {Combinatorial optimization: algorithms and complexity},
      NOTE = {Corrected reprint of the 1982 original},
 PUBLISHER = {Dover Publications, Inc., Mineola, NY},
      YEAR = {1998},
     PAGES = {xvi+496},
      ISBN = {0-486-40258-4},
   MRCLASS = {90-01 (68Q25 90C05 90C27 90C35)},
  MRNUMBER = {1637890},
}

@article {MR3810330,
    AUTHOR = {Dondi, Riccardo and Sikora, Florian},
     TITLE = {Parameterized complexity and approximation issues for the
              colorful components problems},
   JOURNAL = {Theoret. Comput. Sci.},
  FJOURNAL = {Theoretical Computer Science},
    VOLUME = {739},
      YEAR = {2018},
     PAGES = {1--12},
      ISSN = {0304-3975},
   MRCLASS = {05C85 (05C15 68Q25 68R10)},
  MRNUMBER = {3810330},
MRREVIEWER = {Aleksander Vesel},
       DOI = {10.1016/j.tcs.2018.04.044},
       IGNOREURL= {https://doi.org/10.1016/j.tcs.2018.04.044},
}

@article {MR2799278,
    AUTHOR = {Fellows, Michael R. and Fertin, Guillaume and Hermelin, Danny
              and Vialette, St\'{e}phane},
     TITLE = {Upper and lower bounds for finding connected motifs in
              vertex-colored graphs},
   JOURNAL = {J. Comput. System Sci.},
  FJOURNAL = {Journal of Computer and System Sciences},
    VOLUME = {77},
      YEAR = {2011},
    NUMBER = {4},
     PAGES = {799--811},
      ISSN = {0022-0000},
   MRCLASS = {68Q25 (68R10 92D20)},
  MRNUMBER = {2799278},
MRREVIEWER = {Erkki M\"{a}kinen},
       DOI = {10.1016/j.jcss.2010.07.003},
       IGNOREURL= {https://doi.org/10.1016/j.jcss.2010.07.003},
}

@article {MR1986480,
    AUTHOR = {He, G. and Liu, J. and Zhao, C.},
     TITLE = {Approximation algorithms for some graph partitioning problems},
   JOURNAL = {J. Graph Algorithms Appl.},
  FJOURNAL = {Journal of Graph Algorithms and Applications},
    VOLUME = {4},
      YEAR = {2000},
     PAGES = {no. 2, 11},
   MRCLASS = {68R10 (05C85 68W25)},
  MRNUMBER = {1986480},
       DOI = {10.7155/jgaa.00021},
       IGNOREURL= {https://doi.org/10.7155/jgaa.00021},
}

@article {MR1143909,
    AUTHOR = {Hsu, Wen Lian and Tsai, Kuo-Hui},
     TITLE = {Linear time algorithms on circular-arc graphs},
   JOURNAL = {Inform. Process. Lett.},
  FJOURNAL = {Information Processing Letters},
    VOLUME = {40},
      YEAR = {1991},
    NUMBER = {3},
     PAGES = {123--129},
      ISSN = {0020-0190},
   MRCLASS = {68R10 (68Q25)},
  MRNUMBER = {1143909},
       DOI = {10.1016/0020-0190(91)90165-E},
       IGNOREURL= {https://doi.org/10.1016/0020-0190(91)90165-E},
}

@article {MR652906,
    AUTHOR = {Lichtenstein, David},
     TITLE = {Planar formulae and their uses},
   JOURNAL = {SIAM J. Comput.},
  FJOURNAL = {SIAM Journal on Computing},
    VOLUME = {11},
      YEAR = {1982},
    NUMBER = {2},
     PAGES = {329--343},
      ISSN = {0097-5397},
   MRCLASS = {68C25 (05C10)},
  MRNUMBER = {652906},
MRREVIEWER = {Walter Stromquist},
       DOI = {10.1137/0211025},
       IGNOREURL= {https://doi.org/10.1137/0211025},
}

@article {MR3172250,
    AUTHOR = {Marx, D\'{a}niel and Razgon, Igor},
     TITLE = {Fixed-parameter tractability of multicut parameterized by the
              size of the cutset},
   JOURNAL = {SIAM J. Comput.},
  FJOURNAL = {SIAM Journal on Computing},
    VOLUME = {43},
      YEAR = {2014},
    NUMBER = {2},
     PAGES = {355--388},
      ISSN = {0097-5397},
   MRCLASS = {68Q25 (05C85 68Q15)},
  MRNUMBER = {3172250},
MRREVIEWER = {Lenwood S. Heath},
       DOI = {10.1137/110855247},
       IGNOREURL= {https://doi.org/10.1137/110855247},
}

@incollection {MR3836066,
    AUTHOR = {Misra, Neeldhara},
     TITLE = {On the parameterized complexity of colorful components and
              related problems},
 BOOKTITLE = {Combinatorial algorithms},
    SERIES = {Lecture Notes in Comput. Sci.},
    VOLUME = {10979},
     PAGES = {237--249},
 PUBLISHER = {Springer, Cham},
      YEAR = {2018},
   MRCLASS = {68R10 (05C15 68Q25)},
  MRNUMBER = {3836066},
       DOI = {10.1007/978-3-319-94667-2_2},
       IGNOREURL= {https://doi.org/10.1007/978-3-319-94667-2_2},
}

@article {MR3962621,
    AUTHOR = {Pilz, Alexander},
     TITLE = {Planar 3-{SAT} with a clause/variable cycle},
   JOURNAL = {Discrete Math. Theor. Comput. Sci.},
  FJOURNAL = {Discrete Mathematics \& Theoretical Computer Science. DMTCS.},
    VOLUME = {21},
      YEAR = {2019},
    NUMBER = {3},
     PAGES = {Paper No. 18, 20},
   MRCLASS = {68R10 (05C10 05C45 68Q25)},
  MRNUMBER = {3962621},
MRREVIEWER = {Guillermo De Ita},
}

@article {MR2771171,
    AUTHOR = {Dondi, Riccardo and Fertin, Guillaume and Vialette, St\'{e}phane},
     TITLE = {Complexity issues in vertex-colored graph pattern matching},
   JOURNAL = {J. Discrete Algorithms},
  FJOURNAL = {Journal of Discrete Algorithms},
    VOLUME = {9},
      YEAR = {2011},
    NUMBER = {1},
     PAGES = {82--99},
      ISSN = {1570-8667},
   MRCLASS = {68R10 (05C85 68Q25 92D20)},
  MRNUMBER = {2771171},
MRREVIEWER = {Till Tantau},
       DOI = {10.1016/j.jda.2010.09.002},
}

@incollection {MR2506799,
    AUTHOR = {Betzler, Nadja and Fellows, Michael R. and Komusiewicz,
              Christian and Niedermeier, Rolf},
     TITLE = {Parameterized algorithms and hardness results for some graph
              motif problems},
 BOOKTITLE = {Combinatorial pattern matching},
    SERIES = {Lecture Notes in Comput. Sci.},
    VOLUME = {5029},
     PAGES = {31--43},
 PUBLISHER = {Springer, Berlin},
      YEAR = {2008},
   MRCLASS = {68Q25 (68Q17 68W20)},
  MRNUMBER = {2506799},
       DOI = {10.1007/978-3-540-69068-9_6},
}

@inproceedings{zheng2011omg,
    author = {Zheng, Chunfang and Swenson, Krister and Lyons, Eric and
Sankoff, David},
    booktitle = {International Workshop on Algorithms in Bioinformatics},
    doi = {10.1007/978-3-642-23038-7_30},
    organization = {Springer},
    pages = {364--375},
    title = {{OMG}! {O}rthologs in multiple genomes--competing
graph-theoretical formulations},
    year = {2011}
}
\end{document}